\documentclass[sigconf,tikz]{acmart}

\AtBeginDocument{%
  }

\setcopyright{acmcopyright}
\usepackage[english]{babel}
\usepackage[T1]{fontenc}
\usepackage{tikz-uml}

\usetikzlibrary{decorations.pathreplacing,calligraphy}
\tikzstyle{note}=[rectangle,fill,fill=white]
\newcommand{\mathindent}{{\; \; \; \;}}

\usepackage{enumitem}
\usepackage{tikz}
\usepackage{amsmath}
\usepackage{multirow}
\usepackage{amsthm}
\usepackage{mathtools}

\usepackage{amsmath,amsfonts}
\usepackage{algorithmic}
\usepackage{graphicx}
\usepackage[noline,ruled,vlined]{algorithm2e}
\usepackage{subfig}
\usepackage{fontawesome}
\usepackage{pifont}
\usepackage{xurl}

\usepackage[usestackEOL]{stackengine}
\usepackage{colortbl}
\usepackage{array}
\usepackage{pifont}
\usepackage{adjustbox}

\usepackage{acronym}
\acrodef{BF}[BF]{Bloom Filter}
\acrodef{CA}[CA]{Certificate Authority}
\acrodef{TEE}[TEE]{Trust Execution Environment}
\acrodef{CRL}[CRL]{Certificate Revocation List}
\acrodef{PM}[PM]{Pseudonym Manager}
\acrodef{PKI}[PKI]{Public Key Infrastructure}
\acrodef{BU}[BU]{Backward Unlinkability}
\acrodef{VLR}[VLR]{Verifier Local Revocation}
\acrodef{GS}[GS]{Group Signatures}
\acrodef{IBE}[IBE]{Identity Based Encryption}

\newtheorem{assumption}{Assumption}
\newtheorem{theorem}{Theorem}
\newtheorem{observation}{Observation}

\newtheorem{definition}{Definition}
\newtheorem{correctness-argument}{Argument}
\newtheorem{lemma}{Lemma}
\newenvironment{basis}{\noindent\textit{Basis:}}{\hfill$\square$}

\newcommand{\cmark}{ \rotatebox[origin=c]{-90}{ \color{teal}{\ding{51}}} }
\newcommand{\xmark}{ \rotatebox[origin=c]{-90}{ \color{red}{\ding{55}}} }%
\newcommand{\ymark}{ \rotatebox[origin=c]{-90}{ \textcolor[HTML]{FFA700}{\ding{51}}} }%
\newcommand{\tableRotFix}[1]{ \rotatebox[origin=c]{-90}{#1} } 
\newcommand\Mycomb[2][^n]{\prescript{#1\mkern-0.5mu}{}C_{#2}}

\newcommand{\fakedescription}[1]{\medskip\noindent\textit{\textbf{#1}}}

\newcommand{\thesystem}{RRP\xspace}
\newcommand{\theedgesystem}{EDGAR\xspace}
\newcommand{\theabstraction}{Range-Revocable Pseudonyms\xspace}
\newcommand{\theabstractionshort}{RRPs\xspace}
\newcommand{\theabstractionshortsingle}{RRP\xspace}

\begin{document}

\title[Using Range-Revocable Pseudonyms to Provide Backward Unlinkability in the Edge (Extended Version)]{Using Range-Revocable Pseudonyms to Provide Backward Unlinkability in the Edge (Extended Version)*}


\author{Cláudio Correia, Miguel Correia, Luís Rodrigues}

\affiliation{\institution{INESC-ID, Instituto Superior Técnico, Universidade de Lisboa} \country{Portugal}}

\email{{claudio.correia, miguel.p.correia, ler}@tecnico.ulisboa.pt}




\pagestyle{plain}

\begin{abstract}
In this paper we propose a novel abstraction that we have named \theabstraction (\theabstractionshort). \theabstractionshort are a new class of pseudonyms whose validity can be revoked for any time-range within its original validity period. The key feature of \theabstractionshort is that the information provided to revoke a pseudonym for a given time-range cannot be linked with the information provided when using the pseudonym outside the revoked range. We provide an algorithm to implement \theabstractionshort using efficient cryptographic primitives where the space complexity of the pseudonym is constant, regardless of the granularity of the revocation range, and the space complexity of the revocation information only grows logarithmically with the granularity; this makes the use of \theabstractionshort far more efficient than the use of many short-lived pseudonyms.  We have used \theabstractionshort to design \theedgesystem, an access control system for VANET scenarios that offers backward unlinkability. The experimental evaluation of \theedgesystem shows that, when using \theabstractionshort, the revocation can be performed efficiently (even when using time slots as small as 1 second) and that users can authenticate with low latency ($0.5-3.5$ \textit{ms}).
\end{abstract}


\maketitle

\tableofcontents

\section{Introduction}
\label{sec:intro}

Anonymous authentication offers both accountability and privacy, protecting clients from curious application providers while ensuring that only authorized participants are able to use the application~\cite{rahaman2017provably,ishida2018fully,khodaei2018efficient,sun2010efficient}. The number and relevance of applications that require anonymous authentication are increasing. Examples include crowdsensing~\cite{ganti2011mobile,pan2013crowd,ni2017security,bastos2018signature} and Vehicular Networks (VANETs)~\cite{remeli2019automatic,ganan2015epa,mixzone}, where clients voluntarily share information about their environment for the common good. Client authentication is a crucial mechanism to provide accountability for malicious and erroneous activity, ensuring the reliability of these applications~\cite{popa2011privacy,ganti2011mobile}. Unfortunately, authentication can compromise \emph{user privacy}, as it may be associated with sensitive information, such as location~\cite{pan2013crowd}. This is exacerbated by the fact that, in most of these applications, clients are mobile and may need to \emph{authenticate frequently}, e.g., when they move to the range of a different base station or cell. Multiple authentications may be linked to extract additional information such as daily routines~\cite{ganti2011mobile} or health status~\cite{lin2012bewell} for financial gain~\cite{GDPR,christin2016privacy,lauinger2012privacy,at_t,gao_17_656,monetising_car_data}.  Anonymous authentication can be achieved using \ac{GS} schemes~\cite{chaum1991group,boneh2004group,bastos2018signature} or pseudonym certificates~\cite{lysyanskaya1999pseudonym,chaum1985security}. 

A challenging task in this context is to support \emph{revocation} without violating privacy. Revocation aims to prevent some clients from further authenticating in the system. Client revocation may be required in the event of credential misuse, sensor malfunctioning, change in client privileges, stolen secret keys, or when a client leaves voluntarily. Client revocation can be implemented in different ways. We distinguish two main classes of revocation strategies, namely, \emph{global client revocation} and \emph{verifier local revocation}.

Strategies based on global client revocation require all clients to obtain new credentials (or update their credentials) every time a single client is revoked. Examples of this strategy include Ateniese \textit{et al.}~\cite{ateniese2002quasi} (where the group public key is renewed at each revocation) and Ohara \textit{et al.}~\cite{ohara2019shortening} (where a small public membership message is broadcast at each revocation). These approaches make revocation very onerous in scenarios with many clients (e.g., consider vehicle numbers in VANETS) and impractical in mobile settings, where clients may become temporarily disconnected from the network. 

Strategies based on \ac{VLR}~\cite{boneh2004group,bringer2011backward} do not require that all clients are contacted when a given client is revoked. Instead, only the nodes that perform authentication (often called the signature \emph{verifiers}) have to be informed about the revoked clients~\cite{Vtoken,vespa,khodaei2018secmace,sun2010efficient,ishida2018fully,bastos2018signature}. In systems that use pseudonyms, this involves sending to the verifiers a \ac{CRL} with the pseudonyms of the revoked client. In systems based on group signatures, this involves sending a cryptographic token that can be used to trace the digital signatures of the revoked client.

A problem with both approaches is that, if one or more credentials have been used before revocation, an attacker can cross-check the information used for revocation with the information collected when those pseudonyms were used to break the privacy of the client. Ideally, client revocation should not allow linking credentials that have been used prior to the revocation, a property known as \emph{\ac{BU}}~\cite{haas2011efficient,khodaei2018efficient,ishida2018fully,nakanishi2005verifier}. Previous strategies to provide BU assign credentials that are valid only during a given time slot of a certain duration~\cite{haas2011efficient,khodaei2018efficient,rahaman2017provably,bastos2018signature}. Then, when a client is revoked, only the credentials for future time slots are revoked and no information is disclosed regarding credentials used prior to revocation. One can divide these recent strategies as \ac{GS} with time-bounded keys~\cite{chu2012verifier,emura2017group} or pseudonyms with time slots~\cite{haas2011efficient,khodaei2018efficient}. However, these schemes require the use of revocation lists whose size grows linearly with the number of time slots which, in practice, preclude the use of fine-grain time slots. 

Revoking only the credentials for future time avoids backward linkability but, unfortunately, if time slots are large, it may be unacceptable to let revoked clients continue accessing resources until the current slot expires. For this reason, many systems immediately revoke the credentials for the current time slot, at the expense of exposing the client's privacy during that period~\cite{haas2011efficient,khodaei2018efficient,ishida2018fully,nakanishi2005verifier}. Our new class of pseudonyms supports efficient revocation even when fine-grain time slots are used, avoiding this dilemma.

This paper contributes to designing anonymous authentication systems for edge computing scenarios, like VANETs, that offer revocation and \ac{BU}. We make two key contributions in this context:

\fakedescription{Contribution 1:} We propose a novel abstraction named \theabstraction (\theabstractionshort). \theabstractionshort are pseudonyms that can be revoked for any time-range within their original validity period. Clients hold a number of RRPs that is proportional to the number of authentication actions they need to perform, regardless of the granularity of the linkability window. Each RRP should be used at most once during its lifetime. In runtime, the RRP can be used to generate a capability that is only valid for the specific time slot where the RRP is being used. The key feature of \theabstractionshort is that the information provided to revoke a pseudonym for a given time-range cannot be linked with the information provided when using the pseudonym outside the revoked range. In particular, if a pseudonym is revoked at some point in time, it is impossible for an attacker to find out if that pseudonym has been used before that time. We provide an algorithm to implement \theabstractionshort where the space complexity of the pseudonym is constant, regardless of the granularity of the revocation range, and the space complexity of the revocation information only grows logarithmically with the granularity; this supports the use of fine-grain slots and makes the use of \theabstractionshort far more efficient than the use of many short-lived pseudonyms. We show that \theabstractionshort can be used to solve efficiently the \ac{BU} problem for anonymous authentication.

\fakedescription{Contribution 2:}  We propose an access control system for VANET scenarios~\cite{remeli2019automatic,ganan2015epa,mixzone} that uses \theabstractionshort to offer \ac{BU}. Our access control system, named \theedgesystem, illustrates how one can leverage \theabstractionshort to enforce authentication and revocation. In \theedgesystem, we deploy \ac{PM} servers that run on the edge of the network, serving clients with new \theabstractionshort. Since a \ac{PM} server holds sensitive information, and the edge infrastructure in VANETS is known to be exposed to attacks~\cite{Data_Security_and_Privacy_survey_2018,Security_Privacy_Challenges_2017}, we have designed the \ac{PM} server to be executed with the support of a \ac{TEE}, such as Intel SGX enclaves~\cite{mckeen:13,correia2020omega,nguyen2020duetsgx}. This allows the server to provide new \theabstractionshort to clients without disclosing their identities, even if the untrusted environment is compromised.\medskip

\theedgesystem paves the way for arbitrarily small time slots with minimal overhead. Consider a client who during the day goes to the hospital and before that to a nearby shop. When using \theedgesystem, clients only need a number of pseudonyms proportional to the number of resources they need to access (in this example, 2 resources), and not proportional to the granularity of the time slots. With previous work, if the two events above could occur within 20 minutes of each other, a client would require 72 pseudonyms; if the events could occur within 5 minutes of each other, previous works could require 288 pseudonyms. Also, with \theabstractionshort the cost of revocation is logarithmic with granularity: only 12 credentials would need to be revoked with a 20 minute granularity and only 16 credentials would need to be revoked with a 5 minute granularity.

Our prototype uses SGX-enabled Intel NUC (Next Unit of Computing) nodes that can be configured to serve as verifiers or a \ac{PM} server. We measured the latency experienced by clients. Our results show that \theedgesystem can authenticate clients with low latency, in the order of $0.5-3.5$~\textit{ms}, which satisfies the requirements of most latency-sensitive applications, such as augmented reality~\cite{mangiante2017vr,schmoll2018demonstration,satyanarayanan2017emergence} and safety applications~\cite{karagiannis2011vehicular,mec,Emergency_Vehicle}.  We also compared \theedgesystem with a related work~\cite{haas2011efficient} using a real data set of vehicle traces. We show that \theedgesystem achieves multiple orders of magnitude in storage savings and that the revocation can be performed efficiently, even when using time slots as small as 1 second.

\section{Related Work}
\label{sec:background}

\fakedescription{Backward Unlinkability} has been defined in previous work~\cite{haas2011efficient,khodaei2018efficient,ishida2018fully,nakanishi2005verifier} as follows: \textit{when a revocation occurs, the signatures produced by the client before the revocation interval remain anonymous.} The notion of unlinkability captures the inability of an adversarial server to link a revocation phase of the protocol to any individual signing phase. We are interested in non-blocking approaches such as \ac{VLR}~\cite{boneh2004group,bringer2011backward}, where the credential of non-revoked clients remain valid, and only the verifiers need to be informed about the credentials of revoked clients. As discussed next, the most popular anonymous authentication schemes that offer \ac{BU} are based on \ac{GS} or pseudonyms. In both cases, solutions typically consist of assigning different credentials to different time intervals and then revoking only the credentials for future intervals. Unfortunately, in these previous works, the cost of revocation grows linearly with the granularity of the intervals. Note that if intervals are large, it may be unacceptable to wait for the next interval to revoke the credentials: in this case, it may also be necessary to revoke the credentials for the current interval. Unfortunately, this makes all the credentials used in the current interval vulnerable to being linked.

\fakedescription{Anonymous Blacklisting} is a term used to describe techniques that are able to safeguard the privacy of revoked clients.  Techniques to ensure this goal include the use of pseudonyms~\cite{tsang2009nymble}, group signatures~\cite{slamanig2016linking}, accumulators~\cite{au2008perea}, and zero-knowledge proofs (ZKPs)~\cite{tsang2007blacklistable}. Most systems that aim to offer anonymous blacklisting also aim at offering BU~\cite{henry2011formalizing}, but either use computationally expensive cryptographic operations or also incur a cost that is linear with the granularity of the linkability window. For example, BLAC~\cite{tsang2007blacklistable} relies on inherently computationally expensive ZKPs~\cite{henry2011formalizing}. We avoid the use of ZKPs to implement RRPs due to their high cost; instead, we explore more efficient approaches.

\fakedescription{Accumulators, Symmetric Keys, and IBE:} Cryptographic accumulators~\cite{camenisch2002dynamic} may be vulnerable to linkability~\cite{slamanig2016linking}, i.e., the previously performed operations become linkable when a user is revoked. These solutions also lack VLR since clients need to update their witness at each revocation~\cite{au2008perea,camacho2010impossibility}. Credentials based on symmetric keys~\cite{mixzone} require a high level of trust in the verifier and are susceptible to identity theft if the verifier is compromised~\cite{Octopus}. Nymble~\cite{tsang2009nymble} achieves \ac{BU} but requires a central manager to share a symmetric key with every verifier. Credentials based on \ac{IBE}~\cite{BonehIBE} do not provide anonymity (as they consider the user's identity as the public key) and incur considerable overhead due to expensive cryptographic operations.

\fakedescription{Group Signatures with Time-Bound Keys:} 
\ac{GS}~\cite{chaum1991group} allow different signatures produced by different group members to be verified using a common group public key, achieving anonymity in the set formed by the group members. \ac{GS} schemes have been augmented with mechanisms to support VLR, as suggested by Brickell~\cite{brickell2003efficient} and formalized by Boneh and Shacham~\cite{boneh2004group}. Unfortunately, in this scheme, the revocation is performed by publishing a cryptographic token that links all the signatures produced from a revoked member, compromising the anonymity of signatures produced before the revocation. Nakanishi and Funabiki~\cite{nakanishi2005verifier} extend~\cite{boneh2004group} to offer \ac{BU} while preserving VLR. Their approach divides the time into slots and locks a different secret key for each slot, revoking only the keys for current and future slots. Chu \textit{et al.}~\cite{chu2012verifier} introduce the notion of Time-Bound Keys (TBK) by setting a configurable expiration date in each key, improving the revocation performance in VLR-GS schemes. In recent years, different solutions have been proposed, following a similar path while aiming to reduce the revocation cost and complexity. LBR~\cite{slamanig2016linking} requires a trusted online manager to check revocation. Rahaman \textit{et al.}~\cite{rahaman2017provably} embed pseudoIDs in private key parameters and ties the pseudoID to an epoch, improving the revocation check complexity to $\log(R)$, where $R$ is the size of the revocation list. Emura \textit{et al.}~\cite{emura2017group} propose an efficient solution with a constant signing cost, but clients are required to download expiration information at each time slot. In Sucasas \textit{et al.}~\cite{bastos2018signature}, the authors also achieve \ac{BU}, yet, their solution prevents clients from participating in the same task several times. Ishida \textit{et al.}~\cite{ishida2018fully} leverage a mixture of \ac{IBE} with \ac{GS}, generating \ac{IBE} private keys locked to time slots. However, their revocation is still based on the \ac{GS} private key, also having $\mathcal{O}(R~T)$ (where $T$ is the number of time slots). Despite the interesting properties of \ac{GS} schemes, \ac{GS} solutions are usually complex and some may require heavy cryptography operations (such as ZKPs), resulting in few implementations and deployments in real-world scenarios.

\fakedescription{Pseudonyms with Bounded Time Slots:} 
Client anonymity can also be achieved by using pseudonyms. These can be implemented using a \ac{PKI},  where clients maintain multiple keys to represent pseudonyms~\cite{mixzone,Vtoken,rigazzi2017optimized}. Pseudonym-based solutions also struggle to offer \ac{BU}. Some solutions invalidate global information, forcing clients to renew credentials at each revocation~\cite{camenisch2002dynamic,ganan2015epa}, failing to preserve VLR. V-token~\cite{Vtoken}, IFAL~\cite{ifalV2X} and PRESERVE~\cite{ETSI_103} follow the C2C-CC standard~\cite{ETSI_103}, revoking only the long-term vehicles certificates and letting the pseudonyms expire, also failing the VLR. PrivacyPass~\cite{davidson2018privacy} is an anonymous authentication scheme implemented in Cloudflare CDNs, unfortunately, no revocation technique is presented. PUCA~\cite{forster2014puca} requires the owner of the pseudonym to trigger revocation, letting a misbehaving entity evade revocation. The most common solution is to publish all pseudonyms of the revoked client in a \ac{CRL}~\cite{haas2011efficient,khodaei2018efficient,whyte2013security,sun2010efficient}, respecting VRL but failing \ac{BU}.
The challenge of maintaining the unlinkability of pseudonyms after revocation was first addressed by Haas \textit{et al.}~\cite{haas2011efficient}, followed by Khodaei \textit{et al.}~\cite{khodaei2018efficient}, and implemented in SCMS POC~\cite{whyte2013security} and CAMP~\cite{camp2016security} pilots, supported by Volkswagen, Mazda, and Nissan. These solutions associate pseudonyms with time intervals and revoke only pseudonyms of the current and future intervals. However, all interactions in the current slot can still be linked and an adversary can use the revocation information to break anonymity~\cite{haas2011efficient}.

\fakedescription{Privacy at the Edge:} Edge infrastructures~\cite{mec}, supported by numerous \emph{fog nodes}~\cite{Fog_computing,Cisco,bonomi2012fog}, enable computation near clients. Local authentication within edge resources is vital to meet latency requirements and ensure availability. Privacy is a major concern in VANETs, where vehicles generate and transmit substantial amounts of data. Private companies~\cite{gao_17_656,monetising_car_data,hum_verizon} are exploring ways to monetize user data, often at the expense of privacy, with estimates projecting a worldwide market value of \$750 billion by 2030~\cite{gao_17_656,monetising_car_data}. For example, unethical edge providers~\cite{Tech_Monitor,at_t} may sell user data to insurance companies, that can subsequently tailor insurance plans based on individual driving habits~\cite{NHTSA_v2v}. Car-sharing and rental agencies can exploit user data with the same purpose~\cite{fipa}. An attacker could also gain access to user data in the edge infrastructure~\cite{dark_reading,verizon_leak,china_attack,Info_Security_2,Info_Security,edge_news}, inferring if the certain individual is out of the household or has been attending the hospital~\cite{markey2015tracking,gao_17_656,gao_15_775,monetising_car_data,otonomo,NHTSA_privacy,personal_data}.

To mitigate these problem, pseudonyms are recommended in the GDPR~\cite{GDPR2} and by ETSI~\cite{ETSI_102}, and are a standard practice in various connected vehicle pilot programs of major car manufacturers (ETSI~\cite{ETSI_104}, IEEE~\cite{tls_its}, NHTSA~\cite{NHTSA_v2v}) such as CAMP~\cite{CAMP}, New York City~\cite{its_pilot}, and Canada~\cite{FHWA} pilots. According to a study by ETSI on the use of pseudonyms~\cite{ETSI_103}, frequent pseudonym changes enhance privacy: ``\textit{the more often an ITS-S\footnote{Intelligent Transport Systems (ITS) refer to network components, including the On-Board Equipment (OBE) of a vehicle.} changes its pseudonym, the higher its privacy}". However, revoking access rights for a client using different pseudonyms can compromise anonymity when an adversary leverages the revocation information to link the pseudonyms~\cite{haas2011efficient}. Approaches that mitigate backward linkability by associating pseudonyms with time slots result in increased storage requirements for pseudonyms at the client side and, consequently, in the \ac{CRL}~\cite{haas2011efficient}.

\begin{table}[t]\Huge
      \centering
    \resizebox{1.02\columnwidth}{!}{%
{\fontsize{50}{60}\selectfont
    
  \begin{adjustbox}{angle=90}

  \begingroup
    \renewcommand*{\arraystretch}{1.5}%
    \definecolor{tabred}{RGB}{230,36,0}%
    \newcommand{\myTableRotationAngle}{50}
    \newcommand*{\redtriangle}{\textcolor{tabred}{ TEST }}
    \newcommand*{\headformat}[1]{{{\fontsize{50}{0}\selectfont #1 }}}%
    \newcommand*{\vcorr}{%
      \vadjust{\vspace{-\dp\csname @arstrutbox\endcsname}}%
      \global\let\vcorr\relax
    }%
    \newcommand*{\HeadAux}[1]{%
      \multicolumn{1}{ @{} c @{} }{%
        \vcorr
        \sbox0{\headformat{\strut #1}}%
        \sbox2{\headformat{\textbf{\theabstractionshort/\theedgesystem (Ours)}}}%
        \sbox4{9pt}%
        \sbox6{\rotatebox{\myTableRotationAngle}{\rule{0pt}{\dimexpr\ht0+\dp0\relax}}}%
        \sbox0{\raisebox{.1\dimexpr\dp0-\ht0\relax}[0pt][0pt]{\unhcopy0}}%
        \kern.75\wd4 %
        \rlap{%
          \raisebox{.1\wd4}{\rotatebox{\myTableRotationAngle}{\unhcopy0}}%
        }%
        \kern.25\wd4 %
        \ifx\HeadLine Y%
          \dimen0=\dimexpr\wd2+.5\wd4\relax
          \rlap{\rotatebox{\myTableRotationAngle}{\hbox{\vrule width\dimen0 height .4pt}}}%
        \fi
      }%
    }%
    \newcommand*{\head}[1]{\HeadAux{\global\let\HeadLine=Y#1}}%
    \newcommand*{\headNoLine}[1]{\HeadAux{\global\let\HeadLine=N#1}}%
    \noindent

    \newcommand{\myquad}[1][1]{\hspace*{#1em}\ignorespaces}
    

    \begin{tabular}{%
      c  c c|
      c| c| c| c| c| c| c| c| c| c|      
      c| c| c| c| c| c| c| c| c| c|
      c| c| c|
      c| c| c|
      c| c| c|
      c
      >{\bfseries}l%
    }%
      
      &&\head{}&

      \head{ \rotatebox[origin=c]{180}{ \textbf{\theabstractionshort/\theedgesystem (Ours)}  } } &
      \head{ \rotatebox[origin=c]{180}{ Haas \textit{et al.}~\cite{haas2011efficient}  } } &
      \head{ \rotatebox[origin=c]{180}{ Khodaei \textit{et al.}~\cite{khodaei2018efficient}  } } &
      \head{ \rotatebox[origin=c]{180}{ SCMS POC~\cite{whyte2013security}  } } &
      \head{ \rotatebox[origin=c]{180}{ PASS~\cite{sun2010efficient}   } } &
      \head{ \rotatebox[origin=c]{180}{ V-token~\cite{Vtoken}  } } &      
      \head{ \rotatebox[origin=c]{180}{ IFAL~\cite{ifalV2X}  } } &    
      \head{ \rotatebox[origin=c]{180}{ PRESERVE~\cite{ETSI_103} } } &
      \head{ \rotatebox[origin=c]{180}{ Rigazzi \textit{et al.}~\cite{rigazzi2017optimized} } }
      
      &\head{}&
      
      \head{ \rotatebox[origin=c]{180}{  Ishida \textit{et al.}~\cite{ishida2018fully}  } } &
      \head{ \rotatebox[origin=c]{180}{  Nakanishi \textit{et al.}~\cite{nakanishi2005verifier}  } } &
      \head{ \rotatebox[origin=c]{180}{  Rahaman \textit{et al.}~\cite{rahaman2017provably}  } } &
      \head{ \rotatebox[origin=c]{180}{  Sucasas \textit{et al.}~\cite{bastos2018signature}  } } &
      \head{ \rotatebox[origin=c]{180}{  Emura \textit{et al.}~\cite{emura2017group}   } } &
      \head{ \rotatebox[origin=c]{180}{  Ohara \textit{et al.}~\cite{ohara2019shortening}  } } &    
      \head{ \rotatebox[origin=c]{180}{ LBR~\cite{slamanig2016linking} } } &
      \head{ \rotatebox[origin=c]{180}{ Chu \textit{et al.}~\cite{chu2012verifier} } } &
      \head{ \rotatebox[origin=c]{180}{ Bringer \textit{et al.}~\cite{bringer2011backward} } }
      
      &\head{}&
      
      \head{ \rotatebox[origin=c]{180}{ Echeverría \textit{et al.}~\cite{tactical_cloudlets} } }&
      \head{ \rotatebox[origin=c]{180}{ Boneh \textit{et al.}~\cite{BonehIBE} } }
      
      &\head{}&

      \head{ \rotatebox[origin=c]{180}{ Nymble~\cite{tsang2009nymble} }   }  &
      \head{ \rotatebox[origin=c]{180}{ Mix Zones~\cite{freudiger2007mix}  } }&
      \head{ \rotatebox[origin=c]{180}{ Octopus~\cite{Octopus} }   }  
      
      &\head{}&

      \head{ \rotatebox[origin=c]{180}{ Camenisch~\textit{et al.}\cite{camenisch2002dynamic}  } }&
      \head{ \rotatebox[origin=c]{180}{ PEREA~\cite{au2008perea} } }&
      \headNoLine{ \rotatebox[origin=c]{180}{ \Centerstack[t]{ Systems \kern2em     } } }&

      \\
      \sbox0{S}%
      \rule{0pt}{\dimexpr\ht0 + 2ex\relax}%

      & &
      & \cmark  & \cmark  & \cmark  & \cmark  & \cmark  & \xmark & \xmark & \xmark  & \cmark       
      &\cellcolor[HTML]{EFEFEF}       
      & \cmark  & \cmark  & \cmark  & \cmark  & \ymark  & \xmark  & \xmark & \cmark & \cmark 
      & \cellcolor[HTML]{EFEFEF}
      & \cmark & \xmark 
      & \cellcolor[HTML]{EFEFEF}
      & \cmark  & \cmark & \cmark    
      & \cellcolor[HTML]{EFEFEF}
      & \xmark  & \xmark 
      & \cellcolor[HTML]{EFEFEF} & \tableRotFix{ \Centerstack{ VLR\\ }  }
      
      \\
      
      & &
      & \cmark & \cmark & \cmark & \cmark & \cmark & \cmark & \cmark & \cmark & \xmark    
      &\cellcolor[HTML]{EFEFEF}       
      & \cmark & \cmark & \cmark & \cmark & \cmark & \cmark & \cmark & \xmark &  \xmark
      & \cellcolor[HTML]{EFEFEF}
      & \xmark & \cmark
      & \cellcolor[HTML]{EFEFEF}
      & \cmark & \xmark & \xmark
      & \cellcolor[HTML]{EFEFEF}
      & \xmark &  \xmark
      & \cellcolor[HTML]{EFEFEF} & \tableRotFix{ \Centerstack{ BU\\ } }

      \\
      
      & &     
      & \tableRotFix{E-b} & \tableRotFix{S-b} & \tableRotFix{S-b} & \tableRotFix{S-b} & \tableRotFix{S-b} & \tableRotFix{S-b} & \tableRotFix{--} & \tableRotFix{--}  & \tableRotFix{S-b}  
      &\cellcolor[HTML]{EFEFEF}       
      & \tableRotFix{S-b} & \tableRotFix{S-b} & \tableRotFix{S-b} & \tableRotFix{S-b} & \tableRotFix{S-b} & \tableRotFix{S-b} & \tableRotFix{--} & \tableRotFix{E-b} &  \tableRotFix{S-b}
      & \cellcolor[HTML]{EFEFEF}
      & \tableRotFix{S-b} & \tableRotFix{--}
      & \cellcolor[HTML]{EFEFEF}
      & \tableRotFix{S-b}  & \tableRotFix{--} & \tableRotFix{--} 
      & \cellcolor[HTML]{EFEFEF}
      & \tableRotFix{--} & \tableRotFix{S-b}
      & \cellcolor[HTML]{EFEFEF} & \tableRotFix{ \Centerstack{ RDS\\ } } 
      
      \\
      
      & &     
      & \tableRotFix{$R~p~\log(T)$} & \tableRotFix{$R~p~T$} & \tableRotFix{$R~p~T$} & \tableRotFix{$R~p~T$} & \tableRotFix{$R~p~T$} & \tableRotFix{$R~p~T$} & \tableRotFix{--} & \tableRotFix{--} &   \tableRotFix{$R~p~T$}
      &\cellcolor[HTML]{EFEFEF}       
      & \tableRotFix{$R~T$} & \tableRotFix{$R~T$} & \tableRotFix{$R~p~T$} & \tableRotFix{$R~p~T$} & \tableRotFix{$R~T$} & \tableRotFix{$R~\log(N/R)$} & \tableRotFix{$R$} & \tableRotFix{$R~\log(T)$} & \tableRotFix{$R~T$}  
      & \cellcolor[HTML]{EFEFEF}
      & \tableRotFix{$R$} & \tableRotFix{--}
      & \cellcolor[HTML]{EFEFEF}
      &  \tableRotFix{$R~p~T$} & \tableRotFix{$R$} &  \tableRotFix{$R$} 
      & \cellcolor[HTML]{EFEFEF}
      & \tableRotFix{$R$} & \tableRotFix{$R~T$}
      & \cellcolor[HTML]{EFEFEF} & \tableRotFix{ \Centerstack{ Revocation\\data size\\for an epoch } }
      
      \\

      & &     
      & \tableRotFix{$1$} & \tableRotFix{$N$} & \tableRotFix{$N$} & \tableRotFix{$N$} & \tableRotFix{$N$} & \tableRotFix{$1$} & \tableRotFix{$N$} & \tableRotFix{$N$}  & \tableRotFix{$N~p$}    
      &\cellcolor[HTML]{EFEFEF}       
      & \tableRotFix{$N$} & \tableRotFix{$N$} & \tableRotFix{$N$} & \tableRotFix{$(p+1)~N$} & \tableRotFix{$N$} & \tableRotFix{$N~\log(N)$} & \tableRotFix{$R$} & \tableRotFix{$N~\log(T)$} &  \tableRotFix{$N$}
      & \cellcolor[HTML]{EFEFEF}
      & \tableRotFix{$N$} &  \tableRotFix{$N$}
      & \cellcolor[HTML]{EFEFEF}
      &  \tableRotFix{$V$} & \tableRotFix{$N$} &  \tableRotFix{$N$} 
      & \cellcolor[HTML]{EFEFEF}
      & \tableRotFix{$N$} &  \tableRotFix{$N$}
      & \cellcolor[HTML]{EFEFEF} & \tableRotFix{ \Centerstack{  PM storage\\(excluding\\revocation data) } }

      \\
      
      & &     
      & \tableRotFix{$p$}  & \tableRotFix{$p~T$}  & \tableRotFix{$p~T$}  & \tableRotFix{$p~T$}  & \tableRotFix{$p~T$}  & \tableRotFix{$p~T$}  & \tableRotFix{$p~T$}  & \tableRotFix{$p~T$}  & \tableRotFix{$p~T$}  
      &\cellcolor[HTML]{EFEFEF}       
      & \tableRotFix{$1$} & \tableRotFix{$1$} & \tableRotFix{$1$} & \tableRotFix{$1$} & \tableRotFix{$1$} & \tableRotFix{$\log(N)$} & \tableRotFix{$1$} & \tableRotFix{$\log(T)$} &  \tableRotFix{$1$}
      & \cellcolor[HTML]{EFEFEF}
      & \tableRotFix{$1$} & \tableRotFix{$1$}
      & \cellcolor[HTML]{EFEFEF}
      & \tableRotFix{$p~T$} & \tableRotFix{$1$} &   \tableRotFix{$1$} 
      & \cellcolor[HTML]{EFEFEF}
      & \tableRotFix{$1$} & \tableRotFix{$p$}
      & \cellcolor[HTML]{EFEFEF} & \tableRotFix{ \Centerstack{Client\\storage for\\an epoch }  }
      
      \\
      
      & &     
      & \tableRotFix{$\log(T)$}  & \tableRotFix{$1$}  & \tableRotFix{$1$}  & \tableRotFix{$1$}  & \tableRotFix{$1$}  & \tableRotFix{$1$}  & \tableRotFix{$1$}  & \tableRotFix{$1$}  & \tableRotFix{$1$}  
      &\cellcolor[HTML]{EFEFEF}       
      & \tableRotFix{$1$} & \tableRotFix{$1$} & \tableRotFix{$1$} & \tableRotFix{$1$} & \tableRotFix{$1$} & \tableRotFix{$1$} & \tableRotFix{$1$} & \tableRotFix{$\log(T)$} &  \tableRotFix{$1$}
      & \cellcolor[HTML]{EFEFEF}
      & \tableRotFix{$1$} & \tableRotFix{$1$}
      & \cellcolor[HTML]{EFEFEF}
      & \tableRotFix{$1$} & \tableRotFix{$1$} & \tableRotFix{$1$} 
      & \cellcolor[HTML]{EFEFEF}
      & \tableRotFix{$1$} & \tableRotFix{$p$}
      & \cellcolor[HTML]{EFEFEF} & \tableRotFix{ \Centerstack{Signature\\size or\\verification} }
      
      \\

      \multirow{-24.5}{*}{\rotatebox{-90}{$T$: The number of time slots in one epoch. \myquad[1.8] PM: Pseudonym Manager.  \myquad[5.8] V: Number of verifiers. }} &
      
      \multirow{-24.5}{*}{\rotatebox{-90}{RDS: Revocation Data Structure. \myquad[5.7] S-b: Slot-based.  \myquad[10.2] E-b: Epoch-based. }}&
      
      \multirow{-24.5}{*}{\rotatebox{-90}{$N$: Number of clients. \myquad[10] $R$: Number of revoked clients. \myquad[4.5] $p$:~ Required pseudonyms. }}

      & \tableRotFix{$R~p~\log(T)$} &  \tableRotFix{$R$} & \tableRotFix{$R$} & \tableRotFix{$R$} & \tableRotFix{$R$} & \tableRotFix{--} & \tableRotFix{--} & \tableRotFix{--}  & \tableRotFix{$R~p$} 
      & \multirow{-20}{*}{\rotatebox{-90}{\cellcolor[HTML]{EFEFEF}Public-key Encryption}}       
      & \tableRotFix{$R$} & \tableRotFix{$R$} & \tableRotFix{$R~p$} & \tableRotFix{$R~p$} & \tableRotFix{$R$} & \tableRotFix{$R$} & \tableRotFix{--} & \tableRotFix{$R~\log(T)$} & \tableRotFix{$R$}
      & \multirow{-20}{*}{\rotatebox{-90}{\cellcolor[HTML]{EFEFEF}Group Signatures}} 
      & \tableRotFix{$R$} & \tableRotFix{--}
      & \multirow{-20}{*}{\rotatebox{-90}{\cellcolor[HTML]{EFEFEF}Identity Based Encryption}}
      &  \tableRotFix{$R~p$} & \tableRotFix{$R$} &  \tableRotFix{$R$} 
      & \multirow{-20}{*}{\rotatebox{-90}{\cellcolor[HTML]{EFEFEF}Symmetric Encryption}}
      &  \tableRotFix{$R$} &  \tableRotFix{$R$}
      & \multirow{-20}{*}{\rotatebox{-90}{\cellcolor[HTML]{EFEFEF}Cryptographic Accumulators}}  & \tableRotFix{ 
 \Centerstack{ Revocation\\communication\\cost} }

      \\

    \end{tabular}%

    \kern19.5mm 
  \endgroup

\end{adjustbox}
}}

\caption{Properties and complexity offered by different systems, we omit $\mathcal{O}()$ notation for simplicity.}
\label{table:systems}
\end{table}

\fakedescription{Comparison:} Table~\ref{table:systems} summarizes the differences between the related work. We highlight that, with our \theabstractionshort implementation, the size of the revocation information grows only logarithmically with the number of time slots. Cryptographic accumulators do not offer VLR. Symmetric encryption cannot protect user privacy from the verifier, and \ac{IBE} schemes revoke clients by timeout~\cite{BonehIBE} (failing VLR) or issuing the user identifier in a CRL~\cite{tactical_cloudlets} (breaking anonymity). When using GS, clients can generate multiple unlikable signatures with the same secret key, achieving $\mathcal{O}(1)$ for the client storage. Some of these GS schemes revoke clients without providing BU, by simply publishing a token for all possible signatures~\cite{bringer2011backward,chu2012verifier}. Other solutions cannot provide the VLR property~\cite{slamanig2016linking,ohara2019shortening}. Although GS schemes that offer both BU and VLR simultaneously require the revocation procedure to manage a number of credentials that is proportional to the number of time slots $\mathcal{O}(R~p~T)$~\cite{bastos2018signature,rahaman2017provably} and $\mathcal{O}(R~T)$~\cite{emura2017group,nakanishi2005verifier,ishida2018fully}. Another limitation of GS schemes is that they rely on complex and heavy cryptographic operations, in particular, to support revocation; this can induce large latencies when performing authentication.  PKI based schemes are appealing due to their cryptographic efficiency and wide adoption. Some PKI schemes delay the revocation until all pseudonyms expire~\cite{davidson2018privacy,ETSI_103,ifalV2X,Vtoken}, breaking VLR.  Previous schemes that provide both VLR and BU suffer from the same issue as GS scheme~\cite{sun2010efficient,whyte2013security,khodaei2018efficient,haas2011efficient}, by locking pseudonyms to time slots, they require revocation information that is linear with the time slots, times each pseudonym, $\mathcal{O}(R~p~T)$. In addition, these solutions require the clients to carry pseudonyms for all time slots, imposing a storage burden of $\mathcal{O}(p~T)$.

\section{System Model}

This section presents preliminaries on \theabstractionshort and \theedgesystem, which is an edge authentication system based on \theabstractionshort. In \theedgesystem, to perform authentication, a client first uses an \theabstractionshortsingle to obtain a \textit{capability}. This capability, which is only valid for a given target time slot, is then presented to a verifier. To revoke the use of a \theabstractionshortsingle during a range of time slots, the corresponding capabilities are revoked.  There is a level of indirection between the \theabstractionshort assigned to clients and the capabilities used for authentication and revocation that is the enabler to achieve backward unlinkability. 

\subsection{Entities}
\label{sec:entities}

\theedgesystem is composed of three main types of entities: clients, verifiers, and a (distributed) pseudonym manager service. We follow a nomenclature similar to previous work~\cite{rahaman2017provably,haas2011efficient}.

\fakedescription{Clients:} the application client that generates signatures to perform authentication against any verifier. Clients are the holders of \theabstractionshort that they use to generate capabilities to ensure anonymity. Clients are responsible for renewing their \theabstractionshort when needed. 

\fakedescription{Verifiers:} the component that performs client authentication before granting access to a resource. Verifiers are responsible for checking the capabilities provided by clients before granting access. They are also responsible for updating their state by fetching the list of revoked capabilities from the manager servers.

\fakedescription{Pseudonym Manager (PM):} this component is responsible for providing new \theabstractionshort to clients and, when necessary, revoking capabilities generated from these pseudonyms. PM servers are the only entity capable of accessing the true identity of a client so, in our implementation, they run partially inside a \ac{TEE}, ensuring users' anonymity even if the device is compromised. 

\fakedescription{Administrator:} a trusted entity responsible for adding clients to the system and request PMs to revoke clients. 

\begin{figure}
\centering
\includegraphics[width=1.05\columnwidth]{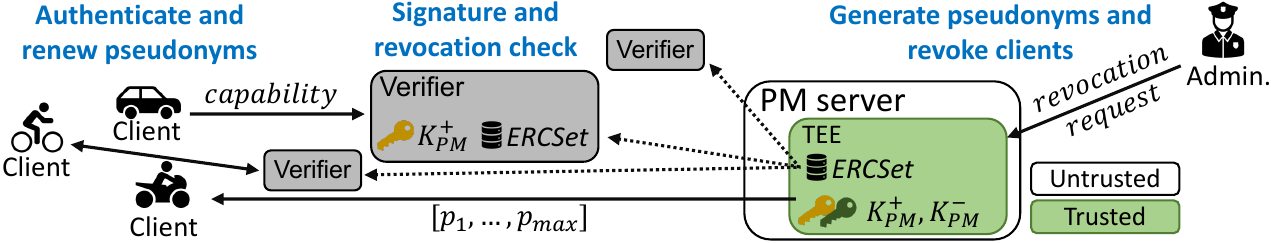}
\captionof{figure}{System Overview.}
\label{fig:simple_overview}
\end{figure}

\medskip

Figure~\ref{fig:simple_overview} shows the interactions between these entities, where the PM uses a \ac{TEE}. The figure represents a typical collective perception scenario, where mobile devices are used to extend human perception. In this example, mobile devices (the clients) authenticate towards the verifiers, to update or download information. When required, clients can contact a nearby PM to renew the set of \theabstractionshort used to generate capabilities. Periodically, verifiers will pull from the PM updated revocation information.

\subsection{Fault Model}

We assume a partial synchrony model~\cite{dls}. In this model there are unstable periods when messages may be arbitrarily delayed, and stable periods when messages between correct entities arrive within at most $\Delta$ units of time. Additionally, we assume that correct processes have access to loosely synchronized clocks, which can differ at most by $\epsilon$. We assume that at most $f$ server nodes can be faulty. We do not place constraints on the number of faulty clients.

Furthermore, verifiers and clients are insecure and prone to Byzantine faults~\cite{byzantinegenerals}. PM servers are executed (partially) inside TEEs and are only subject to crash and omission faults~\cite{Cristian93}. Thus, faulty clients may use expired or invalid credentials when contacting servers, faulty verifiers may arbitrarily deny or grant access to resources, but faulty PMs will never provide faulty information, and will never renew pseudonyms for revoked clients. \theedgesystem ensures liveness during stable periods and offers graceful degradation during unstable periods: when the network is unstable and nodes are unable to receive up-to-date information in a timely manner, they may stop providing service, but never compromise safety.

\begin{table}[]
\centering
\resizebox{\columnwidth}{!}{%
\begin{tabular}{|c|l|c|l|}
\hline
Notation  & Definition                           & Notation             & Definition                                                                                                   \\ \hline
$t$/epoch & Large time interval                  & $s$                  & Time slot, part of an epoch                                                                                  \\
$\delta$  & Duration of a time slot $s$          & $c$                  & Capability                                                                                                   \\
$K^-,K^+$ & Private and public key               & $p$                  & Pseudonym                                                                                                    \\
$e^s$       & Node label for the time slot $s$                    & ERCSet/erc           & Encoded revoked capability sets                                                                              \\
$l_x$  & Latchkeys for $e^s$           & $i$                  & Identifier of pseudonym from a user in an epoch                                                              \\
$h$       & Latchkey tree height                 & $M$                  & Extra pseudonyms to circumvent false positives                                                               \\
\textit{f}         & Number of faulty nodes               & \textit{cid}            & Client identifier                                                             \\
$m$       & Bloom filter size (bits)             & $n$                  & Number of items inserted in a Bloom filter                                                                   \\
$k$       & Number of hash or index functions    & $c_s$                & Number of clients in EDGAR                                                                                   \\
$f_r$     & Fraction of pseudonyms to be revoked & $d$                  & Branching factor of the latchkey tree                                                                        \\
$x$       & False positive rate                  & \multirow{2}{*}{$I$} & \multirow{2}{*}{\begin{tabular}[c]{@{}l@{}}Maximum number of pseudonyms\\ a client can possess\end{tabular}} \\
$N$       & Number of PM replicas                &                      &                                                                                                              \\ \hline
\end{tabular}%
}
\caption{Table of notations.}
\label{table:notation}
\end{table}

\subsection{Threat Model}
\label{sec:ThreatModel}

We trust only administrators and PMs. Following related work~\cite{haas2011efficient,Vtoken}, an administrator is responsible for adding and revoking users in the system by contacting the PM server. We assume that each PM has a processor with \ac{TEE} (e.g., Intel SGX), as shown in Figure~\ref{fig:simple_overview}. All other entities within the system are considered untrusted and susceptible to the control of attackers, potentially engaging in malicious activities. Table~\ref{table:notation} provides the notations.

\fakedescription{Malicious Client:} may attempt to generate pseudonyms or capabilities to impersonate a valid client and access resources to which it is not authorized. It can also try to use old capabilities and pseudonyms after being revoked to authenticate towards verifiers.

\fakedescription{Malicious Verifier:} if a verifier is compromised, the resource that the verifier is protecting becomes unprotected, but this is not the problem we consider in this paper. For example, under a DoS attack, a verifier may be unable to refresh revocation information and should enter a ``safe-mode'' (the safe-mode behaviour is application specific but may be as simple as halting). The problem we consider is that a malicious verifier may try to perform linking attacks~\cite{Vtoken,haas2011efficient}, by associating (linking) different pseudonyms with a single client, breaking user anonymity. This attack becomes trivial when revocation lists that contain all pseudonyms of a client are published~\cite{haas2011efficient}. A malicious verifier may collect all the information/data that it observes, e.g., with the objective of deducing user identity.

\fakedescription{Malicious Pseudonym Manager:} PM code is split in two parts, one that runs inside the \ac{TEE} and one that runs outside the \ac{TEE}. The latter can be compromised and engage in malicious behavior, supporting many of the previously introduced attacks. The untrusted part of PM may attempt to modify, delay, block, or read all messages on the system. This behavior may be done in collusion with other entities to facilitate Linking Attacks or allow a user to evade revocation. Furthermore, we assume that a node suffering from denial of service (DoS) is one of the $f$ faulty nodes and that at most $f$ servers can be faulty.

\fakedescription{Trust Assumptions:} Entities use asymmetric key pairs to establish secure channels. Clients employ \theabstractionshort for authentication, integrity, and non-repudiation. Both the PM and the administrator hold unique key pairs, $(K_{PM}^-, K_{PM}^+)$ and $(K_{admin}^-, K_{admin}^+)$, respectively, being both public keys known to all entities. Specifically, the administrator's public key $K_{admin}^+$ is hard-coded in the enclave's source code. We assume that the PM correctly executes our protocol within the \ac{TEE}, where $K_{PM}^-$ remains securely within the enclave. The PM will only revoke users if instructed by the trusted and authenticated administrator, and will generate fresh pseudonyms for non-revoked and authenticated clients. We assume that there is no collusion between the trusted PM and the verifiers.

The communication between the administrator and the enclave is based on a PKI using their keys. We assume a trusted administrator who only revokes pseudonyms after informing the corresponding clients. Although supporting revocation auditability is beyond the scope of this paper, we discuss different approaches to extend EDGAR and ensure revocation auditability in Section~\ref{sec:RevocationAuditability}. Furthermore, both capabilities and revocation information are accompanied by a digital signature created using $K_{PM}^-$, confirming the origin from the PM \ac{TEE}.

In our work, we make the usual assumptions about the security of \ac{TEE}s/enclaves~\cite{correia2020omega} (code/data executed/stored inside the TEE have integrity and confidentiality guaranteed), about the cryptographic schemes (they satisfy their security properties) and cryptographic keys (secret and private keys are never disclosed). 
In the prototype, we use Ed25519 to generate digital signatures~\cite{bernstein2012high}. 
As a collision-resistant hash function, we use SHA-256.
We use Intel SGX as our \ac{TEE}, although our scheme can be easily adapted to other \ac{TEE}s. We leverage the Intel SGX SDK inside the enclave and OpenSSL outside (all in C/C++).

Although side-channel attacks such as Foreshadow and LVI~\cite{vanbulck2020lvi} exist, we consider the defense from these attacks to be orthogonal to our contribution; possible mitigations are discussed in Bagher\textit{et al.}~\cite{bagher2023sgx}. Correctly synchronizing concurrent data structures can mitigate exploits against synchronization bugs~\cite{weichbrodt2016asyncshock}, with the help of debugging checkers\footnote{In \theedgesystem implementation only a Bloom filter and the current epoch value are accessed concurrently inside the enclave.}~\cite{liu2018d4}.

\section{\theabstraction} 
\label{sec:rrps}

\theabstractionshort are a novel abstraction that provides authentication based on pseudonyms whose validity can be revoked for any time-range within their original validity period. 
Clients hold a number of RRPs that is proportional to the number of authentication actions they need to perform. A validity of an RRP is bounded to an \textit{epoch}. An epoch is divided into time \textit{slots} of length $\delta$. The parameter $\delta$ is application-specific but can be small, e.g., $1$ minute or less.  An epoch is assumed to be much larger than the slot, e.g., $1$ day.  Each RRP should be used for authentication at most once. To perform authentication, a client instantiates a \textit{capability} that is specific to target slot.  If a client is revoked for a time period, pseudonyms are not revoked directly; instead, only the capabilities associated with the time-slots of that period are revoked. We store these capabilities in an \textit{Encoded Revoked Capability Sets} (ERCSet). An \theabstractionshortsingle  can be revoked for a short period, by revoking only the capabilities associated with an interval of time-slots, or permanently, by revoking the capabilities associated with all future time-slots.  Since the revocation information is connected, indirectly, by capabilities, to pseudonyms, when using \thesystem, a client is required to carry a different RRP for each access it needs to perform. However, any given RRP can be used at any slot of the epoch. Thus, the number of RRPs a client needs to keep is independent of the granularity of the time-slots. This contrasts with previous pseudonym-based solutions, where clients need to carry a number of pseudonyms that grows linearly with the \textit{epoch granularity} (i.e., the number of slots in an epoch).

\subsection{Overview}
\label{sec:Overview}

Authentication based on \theabstractionshort uses 3 different related objects, namely (range-revocable) \textit{pseudonyms}, (time-bound) \textit{capabilities}, and ERCSet. At an abstract level, the operations supported by these objects are the following (Appendix~\ref{sec:workflow} offers the operations workflow):

\begin{itemize}[leftmargin=*]
    \medskip\item[--] $p^{epoch}$~$\leftarrow$\textsc{createRRP(}\textit{cid, epoch, $K^-_{PM}$}\textsc{):} used to create a new \theabstractionshort, that can be used by client \textit{cid} during a target \textit{epoch}. Only PMs, using their private key $K^-_{PM}$, can create \theabstractionshort.

    \medskip\item[--] $c^s$~$\leftarrow$\textsc{getCapability(}$p^{epoch}$, $s$, $K^-_p$\textsc{):} used to create a capability associated with an \theabstractionshortsingle $p^{epoch}$ for time slot $s$ ($s$ must belong to the epoch for which the pseudonym was created). Only PMs and the client that owns the pseudonym, and the correspondent private key $K^-_p$, can create capabilities.

    \medskip\item[--] \textit{boolean}~$\leftarrow$\textsc{verifyCapability(}$c^s$, $K^+_{PM}$\textsc{):} To verify if a capability was generated from a valid \theabstractionshort, used by verifiers during authentication, requires the PM public key $K^+_{PM}$.

    \medskip\item[--] \textit{ERCSet}~$\leftarrow$\textsc{createERCSet(}\textit{capabilities}\textsc{):} used only by PMs to create an ERCSet that encodes one or more given capabilities, using some one-way function, such that it is unfeasible to extract a capability from the ERCSet. These capabilities are filtered to ensure that they do not compromise unlinkability (Section~\ref{rrpimplementation}).

    \medskip\item[--] \textit{ERCSet}~$\leftarrow$\textsc{mergeERCSet(}\textit{erc}$_1$, \textit{erc}$_2$\textsc{):} used to merge two ERCSets so that a single ERCSet can be used to capture the revocation of multiple capabilities. PMs and verifiers can merge ERCSets.

    \medskip\item[--] \textit{boolean}~$\leftarrow$\textsc{isRevoked(}\textit{erc}, \textit{capability}\textsc{):} used to verify if a capability is part of an ERCSet. This operation is used by verifiers to check if a capability has been revoked.
\end{itemize}

The manager creates \theabstractionshort on request from authorized clients. If later an \theabstractionshortsingle needs to be revoked for a given range of time slots, the PM generates the corresponding capabilities and encodes them in an ERCSet that is sent to the verifiers.

Clients hold a small number of \theabstractionshort (e.g., corresponding to the number of distinct events), and instantiated a short-lived capability (for the current slot) to authenticate. Then, it presents the capability to the verifier. The verifier checks if the capability is correctly constructed, is \textit{genuine} (i.e., if it was generated from a valid \theabstractionshort) and subsequently check if the capability has not been revoked; only in this case, the client is granted access to the resource. 

To ensure unlinkability, a client must never present two capabilities generated from the same \theabstractionshortsingle, as capabilities generated from the same \theabstractionshortsingle can be linked (cf.~Section~\ref{rrpimplementation}). Therefore, clients have to carry a number of \theabstractionshort proportional to the number of resources they need to access. However, contrary to previous systems, the revocation of an \theabstractionshortsingle for a time-slot does not expose capabilities that may have been used in non-revoked time slots: this is guaranteed by the use of a one-way function to encode revoked capabilities.

\subsection{Making Range-Revocation Efficient}
\label{rangerevocation}

A problem with the use of time-bound pseudonyms is that the number of pseudonyms that need to be revoked grows with the granularity of the time slots. \theabstractionshort are not immune to this problem, because to revoke the use of an \theabstractionshortsingle in a range of time slots, all capabilities associated with those time slots need to be encoded in the ERCSet. However, our implementation of \theabstractionshort uses a mechanism that allows the revocation cost to grow only logarithmically with the granularity, rather than linearly, as previous approaches.

To achieve this goal, a capability is represented by a sequence of \textit{latchkeys}, extracted from a set of latchkeys that are associated with a given RPP. The latchkeys are organized in a tree of fanout $d$, such that there is a leaf latchkey for each individual time slot on an epoch (in this paper, we use $d=2$, i.e., binary latchkey trees). Figure~\ref{fig:tree-v01} provides a simple example where a binary tree of latchkeys is associated with an epoch of 1 hour divided in 4 time-slots of 15 minutes. Note that the latchkey tree structure resembles but is not a Merkle tree~\cite{merkletrees}: the tree nodes are generated independently (the value of a parent node does not depend on the value of its children).

\begin{figure}[th]
  \centering
    \centerline{\includegraphics[width=0.9\columnwidth]{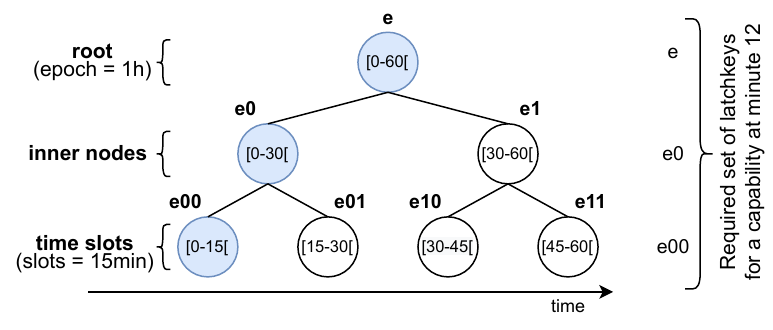}}
  \caption{Latchkeys for time slot $[0,15[$ of epoch $[0,60[$.}
  \label{fig:tree-v01}
\end{figure}

A capability for a given time slot is represented by the set of latchkeys in the path from the root of the tree to the corresponding leaf node in the tree. Using the example of Figure~\ref{fig:tree-v01}, the capability for the first slot would be represented by the following set of latchkeys: $\{ e, e0, e00 \}$. Note that each capability, for each time slot, is always different, because it contains one unique leaf latchkey. However, different capabilities may have some latchkeys in common; in particular, all capabilities include the root latchkey. 

A capability is only considered \textit{valid} if all latchkeys used to represent it are valid. Therefore, the capabilities can be revoked by invalidating any of its latchkeys. In particular, a capability for a given time slot can be revoked by invalidating the leaf latchkey associated with that slot. However, it is also possible to revoke multiple latchkeys by invalidating latchkeys that are inner nodes of the tree: by invalidating an inner node, all the capabilities that are part of the sub-tree rooted at that inner node are invalidated. This can also be illustrated using our example. Consider that the client is revoked at the beginning of the second slot ($e01$) 
until the end of the epoch. At this point, the client may already have used its pseudonym $p$ to generate a capability to access the resource during the first slot. To prevent linkability, the latchkeys used in the first slot cannot be revoked, i.e, the latchkeys $e, e0$ and $e00$ cannot be revealed. To revoke all future capabilities that may be generated with pseudonym $p$, it suffices to revoke latchkeys $e01$ and $e1$. Note that the capabilities generated for the second slot must use $e01$, and the capabilities generated for the third/fourth slots must use $e1$. 

We use this construction to perform revocation efficiently. ERCSets do not explicitly contain capabilities, but only latchkeys that belong to those capabilities and a single latchkey can be used to revoke multiple capabilities. It is easy to show that the number of latchkeys that need to be revoked is at most $\log_d$ with the number of slots. In fact, for each pseudonym valid in the epoch, the number of latchkeys will be given by $\log_d(\textit{granularity})$. 

\subsection{\theabstractionshortsingle Implementation}
\label{rrpimplementation}

We now describe the construction of \theabstractionshort.

\fakedescription{Scheme assumptions:}
We assume that the epoch and time slot size are publicly known to all parties in the system. This means that any party can independently and consistently calculate all the labels from any leaf to the root (i.e., \textit{e}, \textit{$e^0$}, etc.), as illustrated in Figure~\ref{fig:tree-v01}. There is also a maximum number of pseudonyms $I$ that any client can use in any given epoch.

\fakedescription{Cryptographic primitives:} 
We assume there are sources of entropy and a function that allow generating random asymmetric key pairs $(K^-, K^+)$. We assume that there is a function, named $\textsc{DetKeyGen}(\textit{seed})$,  to generate asymmetric key pairs deterministically from a \textit{seed} value. There is also a deterministic signature scheme that, given some private key $K^-$ and a \textit{text} as input, will output a deterministic signature $\textit{sig} = \textsc{DetSign}(K^-, \textit{text})$. The output \textit{sig} can be verified by \textit{true}/\textit{false} = $\textsc{VerSign}(K^+, \textit{text}, \textit{sig})$. Lastly, there is a secure one-way function $\textsc{Digest}(\textit{text})$ that computes a digest on the \textit{text} input and is not possible to invert given the output.

\fakedescription{PM Keys:} There is an asymmetric key pair $(K_{PM}^-, K_{PM}^+)$ associated with every PM. The private key $K_{PM}^-$ is only known by the PMs and is kept in the implementation inside the TEE enclave. The public key $K_{PM}^+$ is known to all participants, including clients and verifiers.

\fakedescription{\theabstractionshort:} An $\textit{RRP}$ is a tuple $\langle \textit{cid}, \textit{epoch}, \textit{i}, K_{p}^-, K_{p}^+, sig_p \rangle$
where \textit{cid} is the client identifier (only known by the client and the PM), \textit{epoch} is the time windows for which the pseudonym is valid,  \textit{i} is a label that can be used to distinguish each pseudonym instance generated for the same epoch, where $i \in [1,I]$. The $(K_{p}^-, K_{p}^+)$ is a unique asymmetric key pair associated with the pseudonym, and $sig_p$ is the signature performed with the private key of the PM over the concatenation of the epoch, and public key of the pseudonym, $\textit{sig}_p = \textsc{DetSign}(K_{PM}^-, \textit{epoch} \mathbin\Vert K_{p}^+)$. Note that some fields of an \theabstractionshortsingle are secrets known only to the client and PM and never revealed to a verifier. In particular, only the client and the PM know the secret key $K_{p}^-$ associated with a given pseudonym. To obtain an \theabstractionshortsingle, a client establishes a secure channel with a PM, presents its client identifier \textit{cid}, and obtains one or more \theabstractionshort for some given target \textit{epoch}. When describing \theedgesystem, we will discuss for which epochs clients are allowed to obtain \theabstractionshort from a PM. 

\fakedescription{Generating $(K_{p}^-, K_{p}^+)$:} The asymmetric key pair associated with a pseudonym is generated using the $\textsc{DetKeyGen}(\textit{$seed$})$ primitive. We use as seed the tuple $\langle \textit{cid}, \textit{epoch}, \textit{i}\rangle$, avoiding the need for the PM to memorize the information associated with all the pseudonyms it created, as it can always re-create them  (as explained below, the key pair is also needed to perform revocation). Recall that \textit{cid} is known only by the client and  the PM. This identifier is securely stored by the PM inside the enclave. Also, \textsc{DetKeyGen} is non-reversible, thus two different public keys created for different epochs and/or instances for the same client cannot be linked with the secret \textit{cid}.

\fakedescription{Latchkeys:} Latchkeys are unique for each pseudonym and are obtained by deterministically signing the label of the corresponding node with the private key $K_{p}^-$ of the pseudonym. Therefore, the latchkey $l_0$ associated with the label node $e^0$ of an $\textit{RRP}$, is generated as $l_0 = \textsc{DetSign}(K_{p}^-, e^0)$, and can be verified by using the public key of the pseudonym by performing $\textsc{VerSign}(K_{p}^+$, $e^0$, $l_0)$.

\fakedescription{Capabilities:} A capability $c$ for a given time slot $s$ is a tuple:

\[ c = \langle K_{p}^+, \textit{sig}_p, l_{leaf},~\ldots,~l_{00},~l_{0},~l_{root}\rangle \]

\noindent where $K_{p}^+$ is the public key of the pseudonym and the latchkeys correspond to the nodes on the path from the root of the latchkey tree to the leaf latchkey node associated with the time-slot $s$. Note that a capability has a number of latchkeys that is logarithmic with the granularity of the time-slots in the epoch. The latchkeys that are part of a capability can be generated on demand, when the capability is created, and are not required to be stored explicitly by the client. It should also be noted that any two capabilities generated from the same \theabstractionshortsingle reveal the same $K_{p}^+$ and can be linked; therefore, a client that wants to prevent authorization request to be linked should always use different \theabstractionshort.

To verify a capability, a verifier performs the following steps. First, it uses the public key of the PM to verify $sig_p$, calculating $\textsc{VerSign}(K_{PM}^+, \textit{epoch}\mathbin\Vert K_{p}^+, \textit{sig}_p)$. Then, it uses $K_{p}^+$ to verify if the latchkeys presented with the capability are in fact associated with that \theabstractionshortsingle, by performing $\textsc{VerSign}(K_{p}^+, e^x, l_x)$. If \textit{all} latchkeys can be verified using $K_{p}^+$ and follow a correct path from the current slot to the root, the capability is genuine. Note that a capability can be genuine but may have been revoked, as explained next.

\fakedescription{ERCSet:} an ERCset in an encoding of a set of latchkeys that represents a set of revoked capabilities. The set of latchkeys encoded in an ERCset has the following properties: \textit{inclusion-of-revoked} -- if a capability has been revoked, at least one of its latchkeys is encoded in the ERCSet; \textit{exclusion-of-non-revoked} -- if a capability has not been revoked, none of its latchkeys are encoded in the ERSet. Below we explain how latchkeys are selected to be encoded in the ERSet to satisfy these properties. Latchkeys are encoded in the ERCSet using a one-way function, $\textsc{Digest}(l_x)$ . Thus, verifiers can check if a given latchkey belongs to an ERCSet but cannot extract latchkeys from the ERCSet. Different data structures that rely on one-way functions could be used to implement ERCSet, including SHA256, or compact data structures such as Cuckoo filters~\cite{fan2014cuckoo}, Cascade filters~\cite{larisch2017crlite} or Count-min sketch~\cite{cormode2005improved}. We use Bloom filters to implement the ERCSet. Bloom filters are efficient and, as discussed later, a good fit for the \theedgesystem architecture. A disadvantage of Bloom filters is that they can present false positives, but we will explain later how \theedgesystem circumvents this limitation.

\fakedescription{Revoking a single capability:} To revoke a capability $c_p$ of a pseudonym $p$, the PM encodes in the ERCSet the leaf latchkey $l_x$ associated with $c_p$. This trivially satisfies the \textit{inclusion-of-revoked} and \textit{exclusion-of-non-revoked} properties: the encoded latchkey belongs to the revoked capability but does not belong to any other capability (each capability is associated with a distinct, unique, leaf latchkey).

\fakedescription{Revoking a range of capabilities:} Revoking a set of capabilities of a pseudonym could be trivially achieved by encoding the corresponding set of leaf latchkeys, but this would have a linear cost. The latchkey hierarchy is used to reduce this cost as follows. Let $d$ be the fanout of the latchkey tree. Any $d$ latchkeys that have the same parent in the latchkey tree can be replaced by their parent. This also satisfies 
the \textit{inclusion-of-revoked} and \textit{exclusion-of-non-revoked} properties of ERCSets: 1) the parent of any latchkey is part of the capability that includes that latchkey and 2) a parent latchkey is not included in capabilities other than the capabilities that include its children. The susbtitution of all $d$ sibling latchkeys by their parent latchkey can be applied recursively in the tree. Note that the root latchkey can only be included in an ERCSet when a pseudonym is revoked for the entire duration of the epoch, because the root latchkey belongs to all capabilities for that epoch. Appendix~\ref{sec:createERCSet} provides a precise description of this algorithm with pseudo-code.

\fakedescription{Merging ERCSet:} An advantage of using Bloom filters is that ERCSet can be easily merged by performing bitwise OR operations. This makes it easy to disseminate revocation lists for many different \theabstractionshort in a single data structure.

\fakedescription{Checking if capability revoked:} Verifiers receive ERCSets from PMs, store them, and use them to check if the capabilities presented by clients have not been revoked. After checking if a capability is genuine, verifiers test if any of the latchkeys are included in the most recent ERCSet. If even a single latchkey is in the ERCSet, the capability is considered revoked.

\subsection{\theabstractionshort Linkability Analysis}
\label{sec:analysis}

A key problem with previous approaches for performing pseudonym revocation is that the information used for revocation could be linked with the information used for access control (in particular, this is obvious when the pseudonym identifier is used both for authentication and revocation). This allows an adversary to collect information about the resources that have been accessed by revoked pseudonyms. When a client has several pseudonyms that are revoked together, an attacker can link the past usage of these pseudonyms to break the privacy of the user. \theabstractionshort avoid this problem because the information used for revocation cannot be linked with the information used for access control. Thus, if a client used one or more pseudonyms prior to revocation, the use of these pseudonyms cannot be linked based on the revocation data. 

Here we present an argument that RRPs offer unlinkability. A more detailed proof is provided in Appendix~\ref{sec:UnlinkabilityProof}.

\begin{observation}
\label{observation:lkgeneration}
Verifiers cannot generate latchkeys associated with a pseudonym.
\end{observation}

\begin{basis}
Latchkeys are generated using the private key $K_{p}^-$ of the pseudonym. The private key is generated using the secret \textit{cid} that is shared between the PM and the client and never revealed to other entities. Therefore, verifiers cannot generate latchkeys.
\end{basis}

\begin{observation}
\label{observation:lkni}
ERCSets do not include latchkeys used outside the revocation interval.
\end{observation}

\begin{basis}
This property is achieved by construction, that ensures the exclusion-of-non-revoked. As described above, when a PM assembles an ERCSet, it never includes in the ERCSet latchkeys that are part of capabilities for time-slots outside the revocation interval.
\end{basis}

\begin{correctness-argument}
\label{teorem1}
Revocation information cannot be linked with authorization information used outside of the revocation interval.
\end{correctness-argument}

\begin{basis}

The revoked latchkeys are encoded in an ERCSet using $\textsc{Digest}(l_x)$, so verifiers cannot extract latchkeys from an ERCSet. Verifiers can only test if a given latchkey has been revoked. However, by Obs.~\ref{observation:lkgeneration}, verifiers cannot generate latchkeys, so they can only test latchkeys that are provided by the client when presenting a capability. By Obs.~\ref{observation:lkni}, latchkeys for capabilities associated with non-revoked time-slots are not included in an ERCSet.\end{basis}

\begin{correctness-argument}
\label{teorem2}
Capabilities generated from different pseudonyms cannot be linked.
\end{correctness-argument}

\begin{basis}
All the information in a capability depends on the asymmetric key pair associated with the pseudonym. Asymmetric key pairs for different pseudonyms are different because they are generated using different seeds (the unique instance number $i$ is part of the seed $\langle \textit{cid}, \textit{epoch}, \textit{i}\rangle$). Additionally, asymmetric key pairs cannot be linked to the seed used for generation (this derives directly from the standard properties of \textsc{DetKeyGen}).
\end{basis}

\section{\theedgesystem}

We now present the design of an anonymous authentication system for the edge that leverages \theabstractionshort to offer backward unlinkability. We have named our system \theedgesystem: \underline{EDG}e distributed \underline{A}ccess cont\underline{R}ol, targeting the VANET scenario~\cite{remeli2019automatic,ganan2015epa,mixzone}. The goal of EDGAR is to reduce the linkability window, improving client privacy at the edge. EDGAR demonstrates how to use our \theabstractionshort abstraction and how to address implementation challenges in a distributed setting.

\subsection{EDGAR in VANETs}
In the VANET scenario, vehicles continuously broadcast CAM messages~\cite{ETSI} containing various information such as their geolocation, sensor readings, direction, and speed. This information is crucial for various edge applications, including enhanced navigation, traffic congestion estimation, remote vehicle diagnostics, autonomous cars, and others~\cite{fipa}. However, as explained in Section~\ref{sec:background}, edge providers can collect and monetize this data, at the expense of users privacy, highlighting the importance for clients to use anonymous authentication methods such as EDGAR. We now contextualize the \theabstractionshort entities to the corresponding entities in \theedgesystem:

\fakedescription{Clients:} These are vehicles that constantly propagate CAM messages with location and sensor readings, with the purpose of enhancing their safety and that of others.

\fakedescription{Verifiers:} Mainly compose by Roadside Units (RSUs)~\cite{correia2022performance} that listen to all CAM messages, aggregate them, and broadcast them in the network. These devices can be deployed by various local entities (e.g., municipal authorities) or edge providers to improve traffic flow, pedestrian safety, and provide services to vehicles such as infotainment or software updates. 

\fakedescription{PM servers:} These are fog nodes with the same services as verifiers but with higher computational capacity and storage, and may be physically more distant than verifiers. Any fog node with TEEs can serve as a PM. We assume the same trust level as mentioned in Section~\ref{sec:ThreatModel}, where the PM does not share its private key or the client pseudonyms. However, the PM is controlled by edge providers, who can access the non-trusted zone outside the TEE.

\fakedescription{Administrator:} In the context of the edge, the administrator should be a trusted entity independent of all applications and providers within the edge. It should work similarly to the current Certificate Authorities (CAs) in PKI.

\begin{figure}[t]
  \centering
    \centerline{\includegraphics[width=0.99\linewidth]{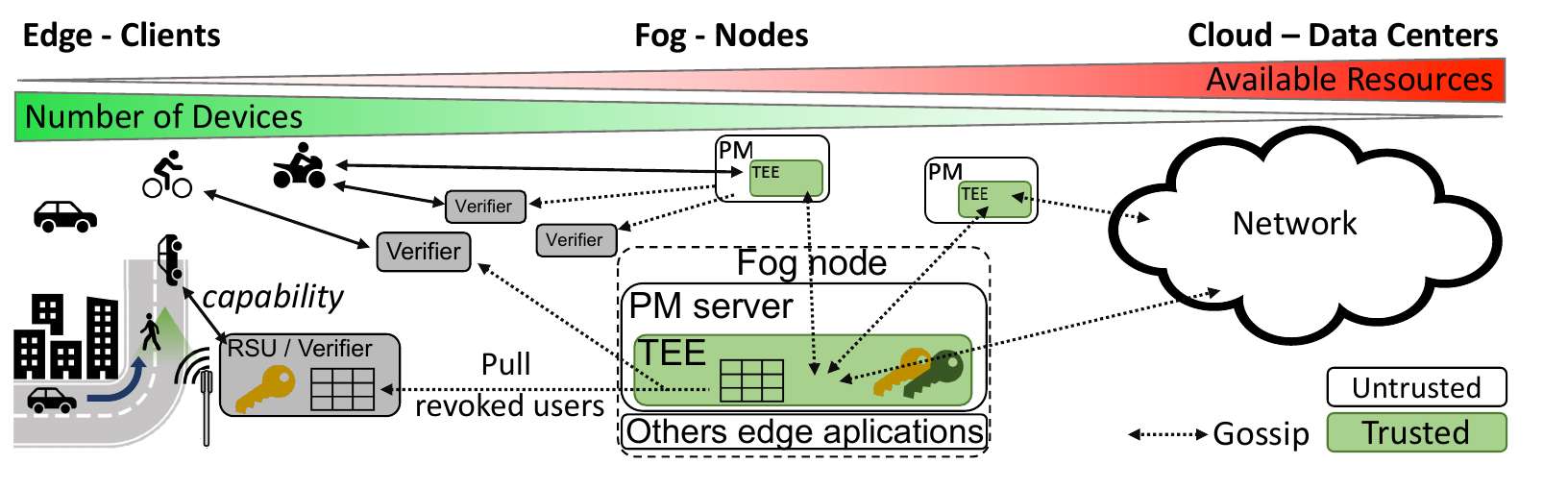}}
  \caption{\theedgesystem entities and VANETs interaction at the edge.}
  \label{fig:usecase_large}
\end{figure}

Figure~\ref{fig:usecase_large} illustrates the interactions in the edge environment. In this example, a vehicle (a client) presents a capability to the RSU (the verifier). If the capability is valid, the RSU accepts the message from the client and alerts the vehicles about pedestrians behind the corner. Authentication is critical to avoid false information that may cause other drivers to break without justification. The RSU relies on multiple nearby fog nodes, running edge replicas of the PM, to update its state. Vehicles can contact a nearby PM replica to renew the pseudonyms used to generate capabilities, if necessary.

\theedgesystem prevents the de-anonymization of clients based on the tracking the RRP usage (i.e., even if edge providers like Verizon, Akamai, and Amazon aggregate data from verifiers and the untrusted zone of the PM). Our pseudonyms also protect users' privacy in case of data leaks from verifiers.

\subsection{Revocation in \theedgesystem}

Although \theabstractionshort supports the revocation of a pseudonym in any range of time slots, in \theedgesystem we assume that clients can be revoked at a target \textit{revocation time slot} (RTS), selected by the administrator, and that all capabilities of that client are revoked for all time slots after RTS (i.e., the revocation range spans from RTS to the end of the epoch). Also, after being revoked, clients become unable to obtain new pseudonyms from edge PMs after some time. In particular, as we will show later, \theedgesystem is able to provide the following guarantee: if a client is revoked in a given epoch $t$, that client may still attempt to use pseudonyms it has obtained for epoch $t+1$ but will not be able to obtain pseudonyms for epoch $t+2$.

\subsection{Epochs, \theabstractionshort, and ERCSet}

Time is divided in epochs and epochs are divided in time slots. The length of an epoch and the granularity $\delta$ of the time slot are application specific. As we show in the evaluation, the efficient revocation mechanism of \theabstractionshort, based on the latchkey hierarchy, supports the use of relatively large epochs and fine-grain granularity, for instance, epochs of one day and time slots of 1 minute. 

There is a limit \textit{I} of the number of pseudonyms that a client can request for a given epoch. When using \theabstractionshort this is not a limitation because clients only need to have a pseudonym for each access regardless of the time slot where the pseudonym is used (and not a different pseudonym for each time slot, as in previous work). Also, clients are only allowed to obtain pseudonyms for the current epoch and for the next epoch (we allow clients to obtain in advance pseudonyms for the next epoch to avoid having PMs to be overload with a rush of requests whenever an epoch begins). This allows us to limit the number of pseudonyms that need to be added to ERCSet when a client is revoked. Also, verifiers only accept requests that use pseudonyms from the current epoch. This allows to garbage collect revocation information from previous epochs safely.

Due to the constraints described above, \theedgesystem is only required to maintain two ERCSets: one associated with the current epoch and another associated with the next epoch. When an epoch $t$ terminates, the ERCSet associated with epoch $t$ can be discarded and a fresh ERCSet is created for the next future epoch ($t+2$).

\subsection{ERCSet dissemination}

The revocation of a client is initiated in the central PM. The PM first generates all possible pseudonyms that the client may have obtained for the current epoch (i.e., by creating the pseudonyms for all instances $1$ \ldots $\textit{I}$) and creates an ERCSet that revokes all the capabilities that may be generated for these pseudonyms in the range starting from the \textit{revocation time slot} (RTS) to the end of the epoch. For this, it uses the algorithm described in Section~\ref{rangerevocation}. It then merges this ERCSet into the global $\textit{ERCSet}_{t}$ for the current epoch. The PM then generates all possible pseudonyms that the client may have obtained for the next epoch and creates an ERCSet that revokes these pseudonyms for the entire epoch (this is very efficient, because it suffices to include the root latchkey of each pseudonym in ERCSet); it then merges this ERCSet in the global $\textit{ERCSet}_{t+1}$.

\fakedescription{Disseminating client revocation among the PMs.} The updated values of $\textit{ERCSet}_{t}$ and $\textit{ERCSet}_{t+1}$ are then disseminated in the system using a two-step procedure. First, they are disseminated from the central PM to all edge replicas of the PM. Then, verifiers pull these values from their nearest PM replicas. \theedgesystem implements the propagation of ERCSets among PM replicas using a gossip-based broadcast protocol. The central PM first selects $f+1$ edge PMs at random and sends them the updated ERCSets. When receiving an ERCSet from another replica, a PM checks if the ERCSet is different from the local version.  If the ERCSet is the same, it discards the redundant update. If the ERCSet is different, it assumes that it may contain new information and merges it with its own ERCSet, picks other $f+1$ edge PMs at random, and sends them the updated ERCSets. This eager push strategy allows revocation information to be propagated quickly on the network. Additionally, a PM that does not receive any updates for more than a predefined gossip timeout engages in pull-gossip with another random PM. Pull gossip is used to recover from temporary crashes or disconnections. A PM that is down when a revocation is eagerly propagated will later obtain the information using pull gossip. Note that the ERCSet for a given epoch always accumulates new information. Thus, any single gossip exchange with an up-to-date server will convey all the information that a node may have missed while disconnected.

\fakedescription{Disseminating latchkey revocations to the verifiers.} The edge will consist of many verifiers placed at different locations. It is not efficient to have all these PM servers sending the same information to all verifiers. Therefore, we only use pull-gossip to propagate ERCSets to each verifier. Each controller periodically picks a PM at random, pulls $\textit{ERCSet}_{t}$ from that server, and merges its content with a local copy of $\textit{ERCSet}_{t}$. If a verifier fails to execute the pull-gossip procedure (possibly due to an adversary jamming the PM), it enters ``safe-mode'' (the specific behavior varies depending on the application, but could involve stopping the service to protect the resource).

\subsection{Obtaining New Pseudonyms}

Clients can obtain pseudonyms from edge PMs. \theedgesystem does not require edge PMs to keep an explicit list of all clients that have been revoked and of their corresponding \textit{revocation time slot}  as this is already encoded in the ERCSet. When requesting new pseudonyms, a client establishes a secure channel with any edge PM and provides its own \textit{cid} and a valid capability. If the client has not been revoked, the PM can provide the requested pseudonyms for the current of for the next epoch. If the client has been revoked, it will be denied access to additional pseudonyms.

To check if a client has not been revoked, a PM performs the following checks: first it verifies if the capability presented by the client is in fact associated to a pseudonym of that client. This procedure leverages that fact that any PM can create all pseudonyms of a client, and therefore, can check if the public key included in the capability corresponds to a public key of one of the valid pseudonyms for that client. Then, it checks if the capability has not been revoked. If the request passes these tests, the PM generates and sends the requested pseudonyms to the client.

\subsection{Size of ERCSets}

\theedgesystem uses Bloom filters to implement ERCSets. Bloom filters have $O(1)$ insertion and query time~\cite{luo2018optimizing}, are space-efficient, and can be merged easily. However, Bloom filters suffer from false positives and should be used with care. In fact, there is evidence that, if not used properly, the false positives generated by Bloom filters can jeopardize the operation of large-scale systems\cite{larisch2017crlite}. We first discuss how the size of the Bloom filters used to implement ERCSets is chosen in \theedgesystem. Later, we discuss how we deal with the fact that false positives cannot be entirely avoided.

The false positive rate of a Bloom filter depends on the filter size $m$ (bits), the number of items to be inserted $n$, and the number $k$ of hash or index functions used for insertion and search. The false positive rate can be approximated as described in~\cite{haas2011efficient}:

\begin{equation}\footnotesize
    P(\textit{false~positive}) =  \left ( 1 - \left ( 1 - \frac{1}{m}\right) ^{kn} \right )^{k}    
    \label{eq:false_positive} 
\end{equation}

In the case of ERCSets,  the average number of items in the Bloom filter is given by:

\begin{equation}\footnotesize
    n =  c_s \times I \times f_r \times \log_d\left ( \frac{\textit{epoch}}{\delta} \right )
    \label{eq:variaveis} 
\end{equation}

where $c_s$ is the number of clients, $I$ is the average number of pseudonyms that each client uses, $f_r$ is the fraction of pseudonyms that may need to be revoked, $d$ is the branching factor of the latchkey tree, \textit{epoch} is the length of an epoch, and $\delta$ is the length of the time slot. By using Eq.~\ref{eq:variaveis} to compute the number of latchkeys that are expected inserted in an epoch in a Bloom filter, we can use Eq.~\ref{eq:false_positive} to select the size of the Bloom filter that limits the probability of having a false positive to some pre-defined threshold.

Let us assume a scenario with $c_s = 250,000,000$ clients (the estimated number of vehicles in the USA), and assume a fraction of pseudonyms that need to be revoked of $10^{-4}$ per year (from~\cite{haas2011efficient}). If we set the length of an epoch $24h$,  this provides an average of revocations per epoch of $f_r = 10^{-4} / 365$. Then, if we set the granularity of the time slots to $\delta =10$ minutes. This yields $144$ time slots per epoch. If we use a binary tree of latch keys, the average number of latchkeys used per revocation is $\log_2(144)$. If we assume that clients need at most $10$ pseudonyms per day, the expected resulting number of items to be added to the Bloom filter is:
\[    n =  250,000,000  \times10  \times (10^{-4}/365) \times \log_2(144) \approx  4911 \]
 
In this scenario, a Bloom filter of $9 KB$ provides a false positive rate of $0.1\%$ (from Eq.~\ref{eq:false_positive}). We can determine the false positive rate of a capability by: $ P_{FP}(C) = 1-(1-x)^h $, where $x$ is the false positive rate of the Bloom filter, from Eq.~\ref{eq:false_positive}, and $h$ is the tree height. In the same scenario, that corresponds to a false positive of 6 in every 1000 capabilities. Additionally, if one wants to increase the granularity of the time slot to $\delta = 1$ minute, the number of time slots per epoch increases 10 times, but the number of latchkeys increases only logarithmically, thus the size of the ERCSet must increase only by a factor of $\log_2(1,440)/ \log_2(144)= 1.46$, i.e, an ERCset of $13 KB$ will be enough to maintain the same false positive rate

\subsection{Circumventing False Positives}

Even if the size of Bloom filters is set appropriately, there is always some probability of the occurrence of false positives. In \theedgesystem we bypass this problem by having clients request $M$ extra pseudonyms, in addition to those that are strictly needed to access the resources. If a false positive occurs, the system automatically picks another unused pseudonym and resubmits the authorization request to the verifier. The only perceived effect by the client is an additional latency in serving the request. We show below that the number of additional pseudonyms that a client needs to carry to circumvent the occurrence of false positives is small. Equation~\ref{eq:full_access} describes the probability that a client will execute all authentication successfully with the help of the $M$ extra pseudonyms.

\begin{equation}\footnotesize
    P(\textit{full~access}) =  1 - \sum_{j=M+1}^{p+M} \Mycomb[p+M]{M+1} \times (  P_{FP}(C)  )^j \times ( 1 - P_{FP}(C)  )^{p+M-j} 
    \label{eq:full_access} 
\end{equation}

When applying Equation~\ref{eq:full_access} to the previous scenario, where $\delta=10$ minutes, and setting $M=0$, it is possible to derive that $1\%$ of the clients may fail some authentication; this number can be reduced to $1.9\cdot10^{-12}$ by setting $M=4$. These extra 4 pseudonyms will require increasing the filter size from $9 KB$ to just $13 KB$ (to achieve the same probability with $M=0$ would require increasing the size of the Bloom filter by $34 KB$).  If each client would require 1000 pseudonyms instead, the probability of a client successfully executing all authentication with $M=0$ is just $63\%$ but when using $M=15$ it increases to $1 -(1.9\cdot10^{-14})$, with a storage increase from $883 KB$ to $896 KB$ (to achieve the same probability with $M=0$  would require a filter of $3.89 MB$).

Leveraging $M$ extra pseudonymous is a space-efficient solution to make the effect of false positives negligible, even for large-scale systems such as \theedgesystem. We could consider alternative encoding techniques that completely eliminate false positives, such as cascade filters~\cite{larisch2017crlite}. However, this would require anticipating all possible false positives in order to create the multiple filter levels; in a scenario of millions of vehicles with multiple pseudonyms, this operation would be very expensive and might become infeasible.

\subsection{Handling Epoch Changes and Quarantine}

When a client is revoked, all future capabilities that can be generated from the pseudonyms it may have obtained are revoked. As discussed above, if a client is revoked in epoch $t$, this requires revoking capabilities for future time slots in epoch $t$ and all the capabilities for epoch $t+1$. Capabilities for epoch $t+2$ and other future epochs do not need to be revoked because \theedgesystem ensures that a client that is revoked in epoch $t$ cannot obtain pseudonyms for epoch $t+2$. This property is guaranteed by a coordination phase that is executed by any PM when it transitions from epoch $t$ to epoch $t+1$. The purpose of the coordination phase is to ensure that any PM that enters in epoch $t+1$ is aware of all revocations performed in epoch $t$ and, therefore, will refuse to issue pseudonyms for epoch $t+2$ to clients that have been revoked in epoch $t$. During coordination, a PM enters in a quarantine mode, where it cannot serve pseudonym requests for epoch $t+2$.

The coordination protocol is implemented by forcing each PM to send to every other PM its version of ERCsets$_t$ and ERCsets$_{t+1}$ at the beginning of the quarantine. Furthermore, a PM waits to receive revocation information from at least $N-f$ PMs before ending the quarantine. Because revocation is performed by updating $f+1$ PMs, and a PM waits for the input of other $N-f$ PMs, for each revoked client, a PM is guaranteed to receive at least one up-to-date ERCset that includes the corresponding revoked capabilities. At the end of the quarantine, a PM is guaranteed to be fully aware of all revocations that have occurred in epoch $t$ and can start serving requests for pseudonyms in epoch $t+2$.

\subsection{Handling a PM failure}

The temporary failure or disconnection of a PM is treated as follows. When the PM server recovers, it will immediately start the pull-gossip procedure. Eventually, it will be able to get up-to-date information on the revoked clients. The same applies to temporarily disconnected PMs. A PM that is offline for a short period of time can operate normally, even if it is slightly outdated. If it is contacted by a verifier, it will not be able to provide the most recent revocation information, but the verifier will be able to fetch that information from another PM in the next gossip interaction. If it is contacted by a revoked client, it may issue new pseudonyms to that client for the current or the next epoch. However, the corresponding latchkeys for those pseudonyms have already been revoked by other PMs, and the client will be revoked in a bounded time.

\subsection{Revocation Auditability}
\label{sec:RevocationAuditability} Informally, revocation auditability refers to a user's ability to verify its revocation status at a service provider before attempting to authenticate. As mentioned in Section~\ref{sec:ThreatModel}, for the current prototype we assume that the administrator informs the clients before proceeding with the revocation. As future work, we plan to augment the system with additional strategies to support revocation auditability. One approach is to use contract-based revocation~\cite{henry2011formalizing}, where the contract semantics are agreed upon by both the user and the provider. This enables the user to determine whether a certain action will constitute misbehavior before deciding whether to engage in it. Another approach has been implemented in Nymble~\cite{tsang2009nymble}. To ensure that fresh revocation information reaches the client, revocation lists must be  published at regular $\Delta i$ time intervals, containing a signature with the corresponding timestamp. When a client communicates with a verifier, it can first request the list, which must have a fresh signature for the current $\Delta i$, and then check if it has not been revoked, otherwise it should halt the authentication process.

\subsection{Discussion}
\label{sec:PropertiesDone}

In this section, we discuss the key features of \theedgesystem.

\fakedescription{Epoch Based Pseudonyms:} Pseudonyms in \theedgesystem are bound to epochs instead of slots. Capabilities are bound to slots, but can be generated at any moment by the client. This decoupling allows clients to store only the desired number of pseudonyms based on application logic, instead of the $\delta$ granularity of slots. 

\begin{figure}[t]
  \centering
    \centerline{\includegraphics[width=\columnwidth]{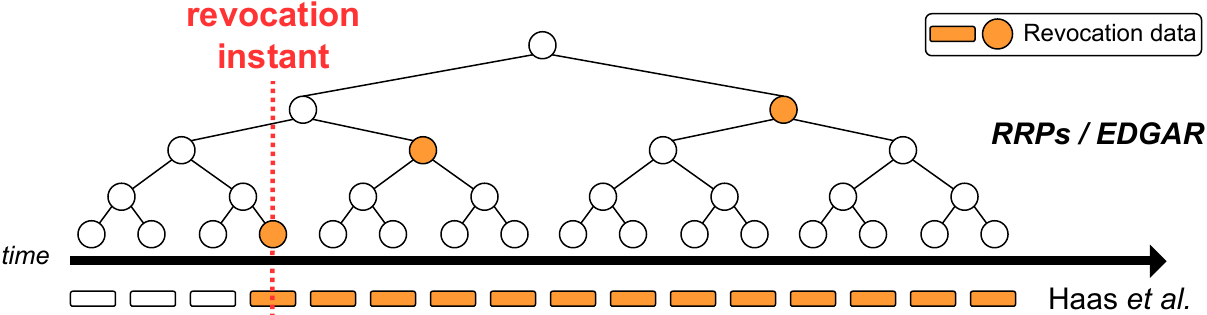}}  
  \caption{Amount of revocation data: \theedgesystem vs~\cite{haas2011efficient}.}
  \label{fig:Comparison}
\end{figure}

\fakedescription{Space Efficiency:} Both edge computing and TEEs have memory constraints, making space efficiency a crucial aspect~\cite{correia2020omega}. With \theedgesystem, revoking a client only requires a logarithmic number of latchkeys, while previous solutions require a linear amount of revocation information relative to the number of time slots (see Figure~\ref{fig:Comparison}).

\fakedescription{Backward Unlinkability:} \theedgesystem revokes future capabilities while ensuring that these capabilities cannot be linked to capabilities used in the past. Unlike previous work, the $\delta$ granularity of the time slots can be arbitrarily reduced without imposing a burden on the system: the number of pseudonyms used by a client does not depend on $\delta$ and the cost of revocation is $\log_2(1/\delta)$.

\fakedescription{Support for Distributed Fault-tolerance:} \theedgesystem distributes and replicates the PM functionality. This increases both availability and resilience. It increases availability because clients can obtain pseudonyms from any correct PM. It increases resiliency, because the coordination required to change epoch effectively prevents PMs that have been isolated or whose clock has been attacked from providing new pseudonyms to revoked clients in future epochs (if the clock is moved backward in time and the server generates invalid pseudonyms for old epochs; if the clock is moved forward in time, the server cannot progress through quarantine).

\fakedescription{Traceability and Accountability:} If required, EDGAR can be extended with mechanisms in which verifiers share (limited) information with the PM to provide traceability and accountability. Specifically, a verifier may present one or more used capabilities to the administrator and ask to revoke or trace the anonymous client that is responsible for such capabilities. Depending on the application and the facts to justify the request, the administrator may agree and forward the capabilities to the trusted central PM. Since the PM can generate all the public keys associated with the pseudonyms it has provided, it can subsequently match the used capabilities with the clients identifiers (note that only the trusted PM can perform this operation; this does not conflict with ensuring unlikability, which aims at preventing non-trusted entities, such as verifiers, from achieving the same goal). However, we have not implemented or evaluated such extensions as part of this work.


\section{Evaluation}
\label{sec:Evaluation}

We evaluate the power of \theabstractionshort, using a prototype of \theedgesystem. We compare the space efficiency of \theedgesystem against a state-of-the-art scheme for BU in the PKI setting. We also show that our scheme offers a latency suitable for edge applications. Finally, we evaluate \theedgesystem's throughput when serving pseudonyms. The source code is available at \url{https://github.com/claudio-correia/RRP-EDGAR}.

We have implemented both a verifier and a PM server on an Intel NUC10i7FNB. An Intel NUC is an example of what a fog node might be, as it possesses modest computational resources but is relatively inexpensive for large-scale deployments. It has an Intel i7-10710U CPU with Intel SGX, 16GB RAM, and Ubuntu 20.04 LTS. We run the Intel SGX SDK Linux 2.13 Release, Intel SSL-SGX~\cite{intelsgxssl} version Linux 2.14\_1.1.1k and OpenSSL 1.1.1k. We used a real-world data set composed of multiple vehicle trajectories in the city of Porto~\cite{dataset_porto}.

\subsection{Space Efficiency}

We have experimentally compared \theedgesystem with Haas~\textit{et al.}, as both support backward unlinkability by dividing the epoch in time intervals in the PKI setting. Haas~\textit{et al.} scheme was more recently implemented in the SCMS POC pilot~\cite{camp2016security} under the name of linkage values technique. The comparison is not trivial since the pseudonyms in Haas \textit{et al.} are locked to a time slot, being invalid if used in any other, while in \thesystem the pseudonyms are free to be used at any moment of an epoch.

For a clear comparison, we test both mechanisms in a real-world use case of a taxi company operating in the city of Porto, using a dataset of taxi trajectories~\cite{dataset_porto}. We choose the mix zones strategy~\cite{mixzone,freudiger2007mix,petit2014pseudonym} for pseudonym changes, i.e., taxis change pseudonyms at crossroads. This use case requires a large number of pseudonyms due to constant vehicle movement, favoring the Haas \textit{et al.} design. In scenarios with fewer pseudonyms needed over the same period, \theedgesystem will outperform Haas \textit{et al.} by even larger margins.

\begin{figure}
  \centering
    \centerline{\includegraphics[width=\columnwidth]{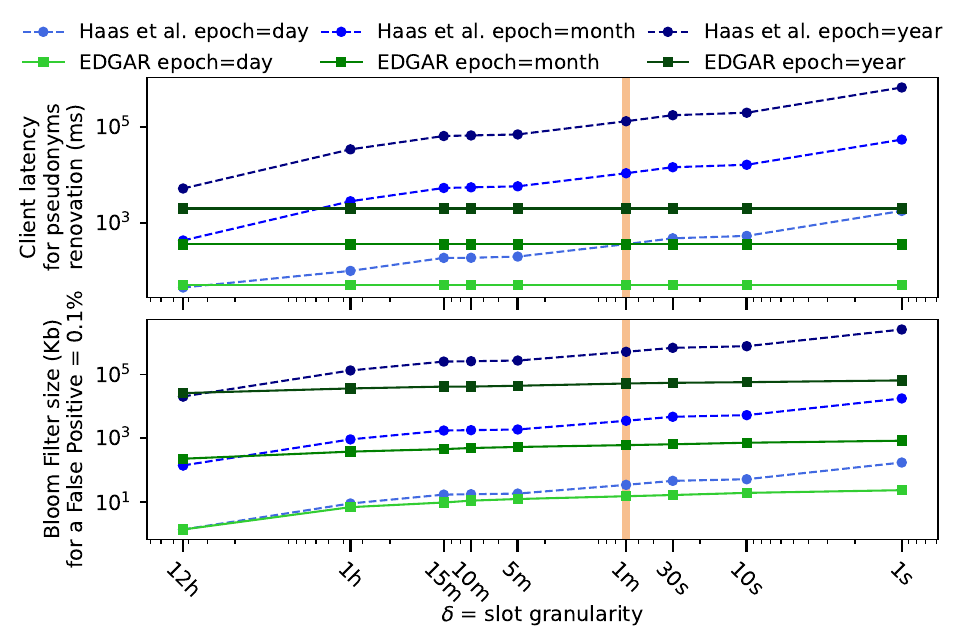}}
  \caption{Storage and latency in Haas \textit{et al.} Vs \theedgesystem.}
  \label{fig:haas_comp}
\end{figure}

Figure~\ref{fig:haas_comp} shows the results obtained in our experiment. The top part of the figure presents the user latency experienced when acquiring all pseudonyms for a specific epoch, and the bottom part presents the required Bloom filter size for each solution. 

In the dataset, we observed that some taxis drive long distances, requiring $692$ pseudonyms per day (which entails $5,677$ per moth and $32,142$ per year). Since \theedgesystem clients can use pseudonyms freely, these values were used as $I$, the number of pseudonyms instances in an epoch. On the other hand, Haas \textit{et al.} must fix the number of pseudonyms in each $\delta$ interval, e.g., with $\delta = 1$ min there was a taxi driver who crossed 12 intersections, forcing to fix 12 different pseudonyms for each $\delta$ interval, which has an explosive effect on the number of pseudonyms generated. Furthermore, as we tighten the $\delta$ interval, the difference between Haas \textit{et al.} and \theedgesystem is more pronounced, for $\delta = 1$ sec and epoch = 1 year we observe an improvement of $\approx 2.5*10^6$KB to $\approx 6.4*10^4$KB, two orders of magnitude lower in the required storage. This results from the logarithmic effect provided by \theabstractionshort, while Haas \textit{et al.} suffers from a linear effect. For large values of $\delta$, when an epoch is divided into a few time slots, Haas \textit{et al.} slightly outperforms \theedgesystem, since the overhead imposed by the latchkey hierarchy is no longer compensated by a significant reduction in pseudonyms.

Finally, we also observed that a large number of taxi trips take around $10$ min., so we believe that $\delta = 1$ min would be a reasonable configuration for this use case, representing a storage saving between 35KB and 15KB (for epoch = 1 day) and 508MB to 51MB (for epoch = 1 year) by implementing \theedgesystem instead of Haas \textit{et al.}, highlighted with the vertical orange line.

\subsection{$\boldsymbol{\delta}$ Granularity vs Latency Trade-off}
\label{sec:latnecy}

\begin{figure}[th]
  \centering
    \centerline{\includegraphics[width=\columnwidth]{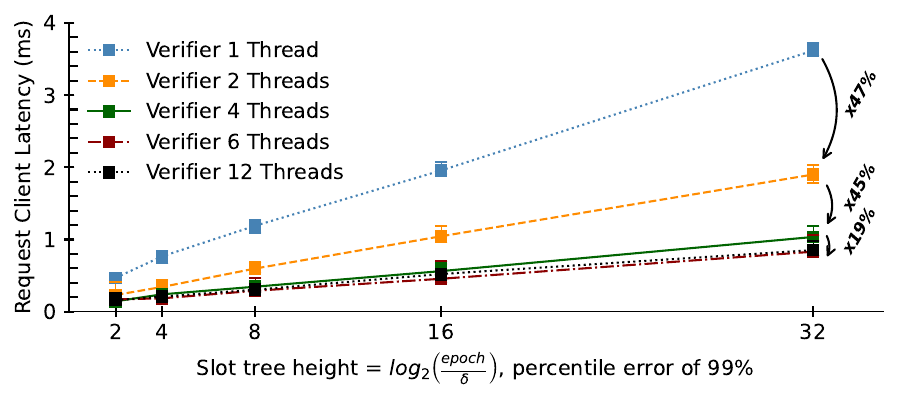}}
  \caption{Client latency for capability verification, for different levels of granularity and number of threads.}
  \label{fig:ACmultihread}
\end{figure}

The use of the latchkey hierarchy makes revocation of a range of time slots efficient, because a single latchkey can be used to revoke many capabilities. The efficiency of revocation comes at the cost of penalty in the authentication procedure, because the verifier must check a number of latchkeys equal to the tree height (instead of verifying only the leaf latchkey). Fortunately, the height of the latchkey tree grows only logarithmically, and latchkeys can be verified in parallel. Therefore, the penalty of \theabstractionshort on latency is small.  Figure~\ref{fig:ACmultihread} shows the latency of the verification procedure as the granularity of the time slots increases.  For instance, a system that uses $\delta = 1$ min  and an epoch of one day requires a binary latchkey tree of depth 11; if the epoch is increased to a month, the depth of the binary tree increases to 16.  In such a case, clients will incur a latency below $2ms$ in a single-thread verifier, still an acceptable latency for edge applications. Note that a tree with 32 levels can support an extreme large number of time slots, such as the ones resulting from using $\delta= 1$ sec and epoch = 70 years; even in this case, clients would experience only $3.5$ ms of latency. We have used multithreading to parallelize the verification of latchkeys, reducing the impact on latency. We observe a latency reduction that is linear with the number of cores of 47\% and 45\% from 1 to 2 threads and from 2 to 4 threads, respectively. Hyperthreading, from 6 to 12 threads, offers no improvement since most of the time is spent in cryptographic operations, and each core has a single ALU.

\subsection{PM Server Throughput} 
\label{sec:isThroughput}

We use the Ed25519 scheme~\cite{bernstein2012high} to obtain deterministic digital signatures, but is not yet available in the SGX SDK. Since the PM is responsible for generating the pseudonyms and runs inside the enclave, we implemented two different versions of the Ed25519 inside the enclave: a portable one~\cite{Ed25519Portable} and one based on OpenSSL~\cite{intelsgxssl}.

The portable version is straightforward to implement in any type of \ac{TEE}, but the lack of optimization affects its performance. In the second implementation, we use the Intel SSL-SGX library to import the OpenSSL library into the enclave. This library was designed for the SGX enclaves, being the most efficient implementation of the scheme. We also evaluated the system with and without SGX, using the OpenSSL library outside the enclave as well.

\begin{figure}[th]
\centering

   \includegraphics[scale=0.49]{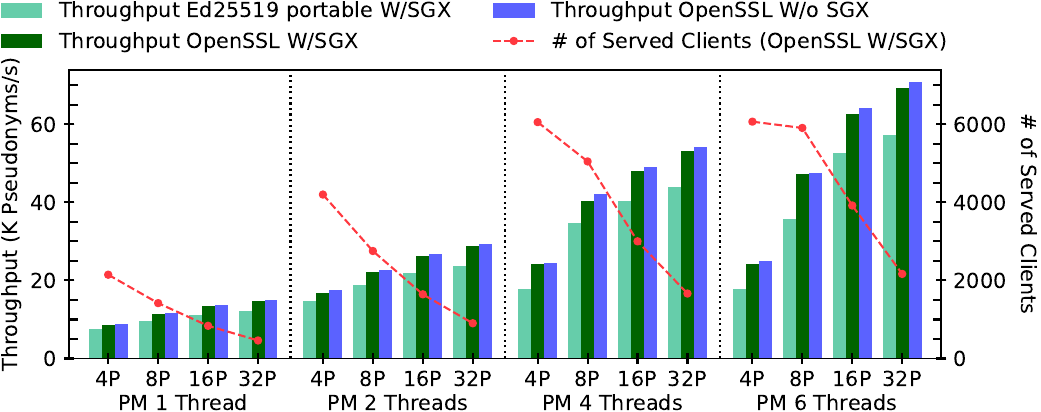}
   \captionof{figure}{Pseudonym throughput and number of served clients by the PM w/ and w/o SGX. $P$ is the number of pseudonyms generated in each request.}
    \label{fig:PM_get_pseudonyms} 

\end{figure}

Figure~\ref{fig:PM_get_pseudonyms} presents the pseudonym throughput for each of these implementations. We vary the number of threads on the PM node side, improving throughput as expected. Additionally, we vary the number of pseudonyms that each client requests from the PM, from 4, 8, 16, and 32 new pseudonyms. We observe that increasing the requested pseudonyms also increases throughput by reducing other system overheads, such as signing and encrypting. On the contrary, the number of served clients decreases since the PM server requires more time for each request. Following the previous discussion and fixing $\delta=1$ min, a taxi could request $32$ new pseudonyms in just $133$ ms, useful for at least the next $32$ min.


\section{Conclusion}
\label{sec:conclusion}

We presented \theabstraction, an abstraction that supports an anonymous authentication scheme based on pseudonyms that is able to enforce backward unlinkability with storage costs that are multiple orders of magnitude lower than those of the related work. The gains derive from our novel technique to decouple pseudonyms from time slots, and authentication/revocation procedures based on the use of latchkeys that can be generated from a given pseudonym for any desired time slot. This technique prevents clients from having to store a large number of unnecessary pseudonyms. As a proof of concept, we have designed and implemented a prototype of \theedgesystem, an authentication system for the edge based on the use of \theabstractionshort. We have used this prototype to perform an experimental evaluation using a real dataset of vehicle traces. The results show that \theedgesystem is capable of offering low latency and storage savings when serving clients. We motivated our work using a VANET scenario, as one of the most prominent use cases for anonymous authentication at the edge computing environment. Still, many other applications, such as crowdsensing, supply chain tracking, augmented reality, etc., also require anonymity and may benefit from our scheme.


\begin{acks}
\small
We thank the anonymous reviewers, Fred B. Schneider, and Carlos Ribeiro for their valuable comments that greatly helped to improve this paper. 

This work was supported by the FCT scholarship 2020.05270.BD, by national funds through Funda\c{c}\~ao para a Ci\^encia e a Tecnologia (FCT) via the INESC-ID grant UIDB/50021/2020 and via the SmartRetail project (ref. C6632206063-00466847) financed by IAPMEI, and by the European Union ACES project, 101093126.

\end{acks}

\bibliographystyle{ACM-Reference-Format}
\bibliography{bibliography.bib}

\appendix

\section{ERCSet Creation Algorithm}
\label{sec:createERCSet}

We now present in Algorithm~\ref{cap_resolution} the pseudo-code for the algorithm used to implement the \textsc{createERCSet} function, previously described using natural language in Section~\ref{sec:Overview}. The algorithm receives multiple capabilities to be revoked (from a given pseudonym). It first merges all the latchkeys from the set of provided capabilities (function \textsc{mergeLatchkeys}). Then, the function \textsc{removeUnsafe} will search and remove any latchkeys that may cover time slots that have not been revoked, to ensure that any authentication information that may be used outside the revocation interval, represented by the provided capabilities, is not present in the ERCSet. Afterward, the function \textsc{removeRedundant} will remove any redundant latchkeys, i.e. any latchkey children, that are already covered by their parent. This step is only to achieve efficient storage since it takes advance of our latchkey tree and only requires a logarithmic number of latchkeys to cover the revoked time slots. Finally, the function \textsc{latchkeysEncoding} takes the remaining latchkeys and inserts each one in a Bloom filter that implements the ERCSet. Our prototype implements an optimized version of this algorithm (the code will be released as open source when the paper is accepted).

\begin{algorithm}[t]
\caption{ERCSet creation, trusted PM side}
\small
\SetAlgoLined
\SetAlgoNoLine%
\label{algorithm:ERCSet}
\SetFuncSty{textsc}
$ {\mathcal C} = \langle C^a, \ldots , C^b \rangle  $: capabilities to be revoked for a pseudonym\\
~\\
\SetKwFunction{FMain}{mergeLatchkeys}
\SetKwProg{Pn}{Function}{:}{}
\Pn{\FMain{${\mathcal C}$}}{
    $\textit{latchkeySet} \gets \emptyset$ // \textit{latchkeySet} is a genuine set (not a multiset) \\  
    \ForEach{ $c_i \in {\mathcal C}$}
    {
        $\textit{latchkeySet} \gets \textit{latchkeySet} \cup \textsc{extractLatchkeys}(c_{i}) $\\
    }
\KwRet \textit{latchkeySet}}
~\\
\SetKwFunction{FMain}{removeUnsafe}
\Pn{\FMain{${\mathcal L}$}}{
    $ \textit{Safe} = {\mathcal L}$\\
    \tcc{Remove latchkeys from non-revoked time slots}

    \ForEach{ $l^i \in {\mathcal L}$}
    {
        \tcc{Returns all children nodes/latchkeys of $l_{i}$}
        $ \textit{descendants}_i \gets \textsc{getLatchkeysSubTree}(l_{i}) $\\   
        \If{ $ 	\exists x \in \textit{descendants}_i \land x \notin {\mathcal L} $ }
        {
            $ \textit{Safe}  \gets  \textit{Safe}  \setminus l_i $\\
        }
    } 
\KwRet \textit{Safe} }
~\\
\SetKwFunction{FMain}{removeRedundant}
\Pn{\FMain{${\mathcal L}$}}{
    $ \textit{NonRedundantSet}  = {\mathcal L}$\\
    \tcc{Remove latchkeys that are covered by parent} 
    \ForEach{ $l_i \in {\mathcal L}$}
    {
        $ \textit{parent}_i \gets \textsc{getParentLatchkey}(l_i) $\\
        \If{ $ \textit{parent}_i \in  {\mathcal L} $ }
        {
            $ \textit{NonRedundantSet} \gets  \textit{NonRedundantSet} \setminus l_i $\\
        } 
    }
\KwRet \textit{NonRedundantSet} }
~\\
\SetKwFunction{FMain}{latchkeysEncoding}
\Pn{\FMain{${\mathcal L}$}}{

    $ \textit{ERCSet} \gets \textsc{createNewBf}()$\\
    \ForEach{ $l_i \in {\mathcal L}$}
    {
        $\textit{ERCSet}.\textsc{bfAdd}(l_{i}) $\\
    }
\KwRet \textit{ERCSet}}
~\\
\SetKwFunction{FMain}{createERCSet}
\SetKwProg{Pn}{Function}{:}{}
\Pn{\FMain{${\mathcal C}$}}{
\textcolor[HTML]{34A222}{\textbf{Enclave Zone Start}}\\
    $\textit{Unfiltered} \gets \textsc{mergeLatchkeys}({\mathcal C})$\\
    $\textit{Safe} \gets \textsc{removeUnsafe}(\textit{Unfiltered})$\\
    $\textit{SafeNonRedundant} \gets \textsc{removeRedundant}(\textit{Safe})$\\
    $\textit{ERCSet} \gets \textsc{latchkeysEncoding}(\textit{SafeNonRedundant})$\\
    \textcolor[HTML]{34A222}{\textbf{Enclave Zone End} }\\
\KwRet \textit{ERCSet}}
\label{cap_resolution}
\end{algorithm}

\section{Unlinkability Proof}
\label{sec:UnlinkabilityProof}

We now provide a proof that how our implementation of Range Revocable Pseudonyms (RRPs) ensures that the information used to revoke clients cannot be linked to the information used by clients when authenticating on non-revoked time-slots. For clarity of exposition and without loss of generality, the proof considers a single epoch, a single client $cid$ with a single pseudonym $p$, and a single Pseudonym Manager (PM).

For convenience we use three sets, \textit{MayHaveBeenUsed}, \textit{Safe} and ERCSet, that
represent different sets of information that are relevant for the correctness of the algorithm.
The set \textit{MayHaveBeenUsed} contains all latchkeys that the client may have provided when authenticating in non-revoked time-slots. 
The set \textit{Safe} contains all latchkeys that are used to perform revocation of a range of time-slots (that corresponds to the output of function \textsc{removeUnsafe} in Algorithm~\ref{cap_resolution}). 
The set ERCSet includes an encoded version of all latchkeys in the set \textit{Safe}, which results from applying a one-way function to each latchkey in \textit{Safe} (that corresponds to the output of function \textsc{latchkeysEncoding} in Algorithm~\ref{cap_resolution}). Note that the function \textsc{removeRedundant} is relevant to reduce the size of the \textit{Safe} but is not relevant to the proof (which remains valid, even if this layer of filtering is not applied).

\textit{MayHaveBeenUsed} captures the information that the client may provide to the verifier when performing authentication. The set ERCSet captures the information included in \textit{Safe}. We consider that \textit{MayHaveBeenUsed} and ERCSet are public in the sense that an adversary can be aware of their content. Only the PM has access to the set \textit{Safe}.

We begin by demonstrating that the sets \textit{MayHaveBeenUsed} and \textit{Safe} are disjoint, i.e., that $\textit{MayHaveBeenUsed} \cap \textit{Safe} = \emptyset$. Then, we show that encoding the latchkeys of \textit{Safe} in the set ERCSet achieves the desirable unlinkability.

\begin{assumption}
\label{assumption:asymmetric}
There is 
a secure deterministic digital signature scheme.
\end{assumption}

We assume the availability of a secure deterministic digital signature scheme, such as Ed25519~\cite{bernstein2012high}, which ensures the usual authentication, integrity, and non-repudiation properties. This scheme supports $\textsc{DetKeyGen}(\textit{seed})$ that creates a key pair $(K^-, K^+)$ from any given seed. Signatures are also deterministic, generated by \textsc{DetSign()} using $K^-$, and verified by \textsc{VerSign()} with $K^+$. Finally, it is not possible to forge a signature without $K^-$. 

\begin{assumption}
\label{assumption:function}
There is a secure one-way function. 
\end{assumption}

We assume the availability of a secure collision-resistant hash function \textit{Digest()}, such as SHA256, that is easy to compute on every input, but not possible to invert given the output.

\begin{assumption}
\label{assumption:cid}
Only the client and the PM can access $cid$.
\end{assumption}

We assume that both the client and the PM will securely store the secret $cid$ and never disclose it.

\begin{definition}
\label{definition:tree}
A time-slot tree is a $n$-ary tree data structure that represents the time period of an epoch divided in time slots.
\end{definition}

A time-slot tree (see Figure~\ref{fig:tree-v01}) is a tree data structure where the time period of an epoch is divided into smaller slots of size $\delta$; each time slot is represented as a leaf node in the tree, identified by a unique node identifier $e^x$. An inner node of the tree has $n$ children and represents the entire time interval of its children nodes;  the tree root node $e$ captures the entire epoch time interval.

\begin{definition}
\label{definition:latchkey}
A valid latchkey is a digital signature of a tree node.
\end{definition}

A latchkey $l_x$ represents the unique bond of a pseudonym $p$ to a node $e^x$ in the tree. This bond is implemented by a digital signature performed over the node identifier (cf.~Definition~\ref{definition:tree}), where $l_x = \textsc{DetSign}(K_{p}^-, e^x)$. By using the private key of the pseudonym ($K_{p}^-$), we enforce authentication, integrity, and non-repudiation of the latchkey (see Assumption~\ref{assumption:asymmetric}).

\begin{definition}
\label{definition:capability}
A valid capability $C_p$ is defined as:
\[ c = \langle K_{p}^+, sig_p,~l_{leaf},~\ldots~l_{root}\rangle \]
\end{definition}

\noindent
where $ K_{p}^+$ is the public key of the pseudonym $p$, and the signature $\textit{sig}_p$ is used to prove the pseudonym authenticity (cf.~Section~\ref{rrpimplementation}), being generated by the PM as $\textit{sig}_p = \textsc{DetSign}(K_{PM}^-, \textit{epoch}\mathbin\Vert K_{p}^+)$, $\mathbin\Vert$ represents the concatenation of both values. Finally, $l_{leaf},~\ldots~l_{root}$ is a sequence of valid latchkeys (from Definition~\ref{definition:latchkey}) associated with the nodes in a path from a leaf node to the root node of the time-slot tree. The path is required to perform pseudonym revocation efficiently. Each time-slot $s$ is associated with a different capability, which is uniquely defined by the path from the root node of the time-slot tree to the leaf node associated to time-slot $s$.

\begin{definition}
\label{definition:NRCSet}
The Non-Revoked Capability Set is the set of capabilities for all non-revoked time-slots. 
\end{definition}

Note that the Non-Revoked Capability Set is never used in the algorithm (in fact, a client uses each pseudonym a single time, and never for multiple time-slots). The Non-Revoked Capability Set is an abstract artifact that is useful for the proof; it represents all capabilities that \underline{may} be used by a client in non-revoked time lots.

\begin{definition}
\label{definition:sul}
We define \textit{MayHaveBeenUsed} as the union of all the latchkeys in all capabilities from the Non-Revoked Capability Set.
\end{definition}

Like the Non-Revoked Capability Set, \textit{MayHaveBeenUsed} is never explicitly used in the algorithm and can never be collected by verifiers because, as noted before, a correct client only uses each pseudonym once (for a single time-slot). The set \textit{MayHaveBeenUsed} captures all latchkeys that \underline{may} be exposed to verifiers by the client when authenticating in non-revoked time lots.

\begin{definition}
\label{definition:RCSet}
The Revoked Capability Set  is the set of all capabilities for all revoked time-slots.
\end{definition}

The Revoked Capability Set is provided as input for Algorithm~\ref{cap_resolution}, that constructs an ERCSet that is subsequently used for revocation.


\begin{definition}
\label{definition:usrl}
The Unfiltered Revoked Latchkey Set,  or simply \textit{Unfiltered}, is the union of the latchkeys in all capabilities in the Revoked Capability Set.
\end{definition}

The  \textit{Unfiltered} set is obtained by applying the  \textsc{mergeMatchkeys} function to the Revoked Capability Set in Algorithm~\ref{algorithm:ERCSet}.

\begin{definition}
\label{definition:filter}
The \textit{Safe} set is is constructed by applying the \textsc{removeUnsafe} function from Algorithm~\ref{algorithm:ERCSet} to the \textit{Unfiltered} set.
\end{definition}

\begin{definition}
\label{definition:encoded}
The ERCSet is constructed by encoding all the latchkeys in the \textit{Safe} set with one or more one-way hash functions.
\end{definition}

In our implementation we use Bloom filters to implement the ERCSet. The one-way hash functions are implemented using \textit{Digest($~l_x$)} from Assumption~\ref{assumption:function}.

\begin{lemma}
\label{lemma:disjoint}
\textit{MayHaveBeenUsed} and \textit{Safe} are disjoint.
\end{lemma}

\begin{proof} We now show that $\textit{MayHaveBeenUsed} \cap \textit{Safe} = \emptyset$ by contradiction:
The same latchkey being in both sets can be written as: $\exists x: x \in \textit{MayHaveBeenUsed} \cap \textit{Safe}$. 

Let us assume that the latchkey $x$ belongs to \textit{MayHaveBeenUsed} then, from Definition~\ref{definition:NRCSet}, there is a capability $c_s$, of a non-revoked time slot $s$, that contains $x$. Let $l_s$ be the latchkey associated with the leaf node for time slot $s$. Note also that if $x$ is associated with an inner node of the time-slot tree, the leaf node for time slot $s$ is a descendant of that inner node, because a capability includes all nodes in the path from the root to the leaf node, and nodes in the path have a parent-child relation. 

Let us assume that the latchkey $x$ belongs to \textit{Safe}. Then $l_s$ must belong to \textit{Unfiltered}, otherwise, because $l_s$ descends from $x$, function \textsc{removeUnsafe} in Algorithm~\ref{algorithm:ERCSet} would remove $x$ from \textit{Safe} if $l_s \not\in \textit{Unfiltered}$.

Thus, if $\exists x: x \in \textit{MayHaveBeenUsed} \cap \textit{Safe}$ then $l_s$ must belong both to \textit{MayHaveBeenUsed} and to \textit{Unfiltered}. Because $l_s$ is associated with a leaft node, this means that there is a time slot that is both revoked and non-revoked, a contradiction.
\end{proof}

\begin{lemma}
\label{lemma:latchkeys}
Only the client and the PM are capable of generating a valid capability $C_p$ and the respective latchkeys.
\end{lemma}

\begin{proof} From Definition~\ref{definition:capability}, capability $C_p$ includes a set of latchkeys. From Definition~\ref{definition:latchkey}, to construct each latchkey, access to the pseudonym private key $K_{p}^-$ is required. From Assumption~\ref{assumption:asymmetric}, the key pair $(K_{p}^-, K_{p}^+)$ for the pseudonym $p$ that the client holds is generated using the $cid$ as seed in the asymmetric scheme. From Assumption~\ref{assumption:cid} only the client and PM can access the $cid$ and, therefore, only the client and the PM can generate $K_{p}^-$. Therefore, only the client and the PM can generate valid latchkeys for the capability $C_p$.
\end{proof}

\begin{lemma}
\label{lemma:extract}
An adversary can only access latchkeys from \textit{MayHaveBeenUsed}.
\end{lemma}

\begin{proof} \textit{MayHaveBeenUsed} contains all the  latchkeys that appear in non-revoked capabilities. An user, to authenticate in a non-revoked time-slot, must generate and present the corresponding capability to a verifier (that is non-trusted). Therefore, the attacker may have access to latchkeys in \textit{MayHaveBeenUsed}. On the other hand, by Definition~\ref{definition:encoded}, ERCSet is constructed by encoding each latchkey $l_x$ using $\textsc{Digest}(l_x)$, and by Assumption~\ref{assumption:function} the adversary is not capable of inverting any entry in ERCSet to extract a latchkey.
\end{proof}

\begin{theorem}
\label{theorem:unlinkability}
The information used to revoke clients cannot be linked to the information used by clients when authenticating on non-revoked time-slots.
\end{theorem}

\begin{proof}
Revocation is performed using exclusively latchkeys from the \textit{Safe} set that are encoded using \textsc{Digest()}. To link the information provided during authentication with the information used for revocation, the adversary would need to have access to at least a latchkey $u \in \textit{MayHaveBeenUsed}$ and to a latchkey $r \in \textit{Safe}$. From Lemma~\ref{lemma:extract}, the adversary can only access the latchkeys in \textit{MayHaveBeenUsed}. From Lemma~\ref{lemma:disjoint}, \textit{MayHaveBeenUsed} and \textit{Safe} are disjoint. 
Therefore, the information used to revoke clients (latchkeys from the \textit{Safe} set) cannot be linked to the information used by clients when authenticating (the latchkeys in \textit{MayHaveBeenUsed}).
\end{proof}

From Theorem~\ref{theorem:unlinkability}, a client using a single pseudonym obtains unlinkability guarantees, since no revocation information can be linked with authorization information used outside the revocation interval. Next we show that RRPs also offer unlinkability when multiple pseudonyms are used.

\begin{theorem}
\label{theorem:capdiff}
Capabilities generated from different pseudonyms cannot be linked.
\end{theorem}

\begin{proof}
From Definition~\ref{definition:capability}, all the information in a capability depends on the asymmetric key pair associated with the pseudonym. From Assumption~\ref{assumption:function}, asymmetric key pairs for different pseudonyms are different because they are generated using different seeds (the unique instance number $i$ is part of the seed $\langle \textit{cid}, \textit{epoch}, \textit{i}\rangle$). Additionally, asymmetric key pairs cannot be linked to the seed used for generation: this is guaranteed by the security of the key generation scheme. 
\end{proof}


\section{RRP Protocol Workflow}
\label{sec:workflow}


\begin {figure*}
\centering
\begin{adjustbox}{width=\textwidth}

\begin{tikzpicture}
  \begin{umlseqdiag}
      \umlobject[x=0,fill=white,no ddots]{Client}
      \umlobject[x=5,fill=white,no ddots]{Verifier}
      \umlobject[x=14,fill=white,no ddots]{Pseudonym Manager}
      \umlobject[x=20.5,fill=white,no ddots]{Administrator}
  
      \begin{umlcall}[dt=5,padding=10,op=\textsc{requestRRP(\textit{cid})}, return={$\langle sig_{p1}, ..., sig_{pI} \rangle$}]{Client}{Pseudonym Manager}
      \end{umlcall}

        \begin{umlcallself}[dt=3,padding=8]{Client}
        \end{umlcallself}

        \begin{umlcallself}[dt=3,padding=4.5]{Client}
        \end{umlcallself}
        
        \begin{umlcall}[dt=0,padding=25,op={$\textsc{sendRequest}(req, sig_{req}, c)$}, return={accepted}]{Client}{Verifier}
        \end{umlcall}

      \begin{umlcall}[dt=47,padding=27,op={$\textsc{revoke}(cid, \langle s, ..., s^n \rangle)$}, return={return}]{Administrator}{Pseudonym Manager}
      \end{umlcall}
  
      \umlsdnode[dt=3]{Client}
      \umlsdnode[dt=3]{Verifier}
      \umlsdnode[dt=8]{Pseudonym Manager}
      \umlsdnode[dt=8]{Administrator}
  \end{umlseqdiag}
  
  \node[note,align = left] at (16.5,-1.5) {
    \small $\textit{epoch} = \textsc{currentEpoch}()$ \\
    \small for $i$ in $I$: \\ 
    \small $\mathindent \textit{seed} = \textsc{Digest}(\textit{cid, epoch}, i)$ \\
    \small $\mathindent K_{p}^{-},K_{p}^{+} = \textsc{DetKeyGen}(\textit{seed})$ \\
    \small $\mathindent sig_{p} = \textsc{DetSign}(K_{PM}^{-}, \langle \textit{epoch}, K_{p}^{+} \rangle)$
  };

  \draw[decorate, decoration = {brace, amplitude=5pt}] (19,-0.5) --  (19,-2.5);
  \node[note, align = left, rotate=90] at (19.5,-1.9) {\textsc{createRRP}$(cid, epoch)$};

  \node[note,align = left] at (2.396,-4) {
    \small $s = \textsc{currentSlot()}$ \\
    \small $\textit{seed} = \textsc{Digest}(\langle \textit{cid, epoch, i} \rangle)$ \\
    \small $K_{p}^{-}, K_{p}^{+} = \textsc{DetKeyGen}(seed)$ \\
    \small $set_{latchkeys\_labels} = \textsc{path}(\textit{s, epoch})$ \\
    \small for $e$ in $set_{latchkeys\_labels}$: \\
    \small $\mathindent l = \textsc{DetSign}(K_{p}^{-}, e)$ \\
    \small $c = \langle K_{p}^{+}, sig_p, l, ..., l_n \rangle$
  };

  \draw[decorate, decoration = {brace, amplitude=5pt}] (4.8,-2.6) --  (4.8,-5.4);
  \node[note, align = left, rotate=90] at (5.3,-4) {\textsc{getCapability}$(sig_{p}, s)$};

  \node[note,align = left] at (1.888,-6.4) {
    \small $req = ``open"$ \\
    \small $sig_{req} = \textsc{DetSign}(K_{p}^{-}, req)$ \\
  };

    \node[note,align = left] at (-0.0302,-7.4) {
    \small \\
  };

    \node[note,align = left] at (0.0305,-7.4) {
    \small \\
  };
  
  \node[note,align = left] at (7.4,-8.6) {
    \small $\textsc{VerSign}(K_{p}^{+}, req, sig_{req})$ \\
    \small $\textit{s} = \textsc{currentSlot}()$ \\
    \small $\textit{epoch} = \textsc{currentEpoch}()$ \\
    \small $\textsc{VerSign}(K_{PM}^{+}, \langle \textit{epoch}, K_{p}^{+} \rangle, sig_{p})$ \\
    \small $set_{latchkeys\_labels} = \textsc{path}(s, \textit{epoch})$ \\
    \small for $e$ in $set_{latchkeys\_labels}$: \\
    \small $\mathindent \textsc{VerSign}(K_{PM}^{+}, e, l)$
  };

  \draw[decorate, decoration = {brace, amplitude=5pt}] (10,-7.2) --  (10,-10);
  \node[note, align = left, rotate=90] at (10.5,-8.5) {\textsc{verifyCapability}$(c)$};

  \node[note,align = left] at (6.8,-10.9) {
    \small $set_{latchkeys} = \langle l, ..., l_n \rangle$ \\
    \small for $l$ in $set_{latchkeys}$: \\
    \small $\mathindent \textit{encoded} = \textsc{Digest}(l)$ \\
    \small $\mathindent \textsc{checkSet}(\textit{erc, encoded})$
  };

  \draw[decorate, decoration = {brace, amplitude=5pt}] (8.5,-10) --  (8.5,-11.8);
  \node[note, align = left, rotate=90] at (9,-10.8) {\textsc{isRevoked}$(\textit{erc}, c)$};

  \node[note,align = left] at (16.25,-8.3) {
    \small $erc_{tmp} = \emptyset$ \\
    \small for $i$ in $[ 1, I ]$: \\
    \small $\mathindent c = \emptyset$ \\
    \small $\mathindent \mathrm{for} \; s \; \mathrm{in} \; \langle s, ..., s^n \rangle$: \\
    \small $\mathindent \mathindent c = c \cup \textsc{getCapability}(\_, s)$ \\
    \small $\mathindent \mathrm{\# defined \; at \; Algorithm \; 1}$ \\
    \small $\mathindent erc_{tmp} = \textsc{createERCSet}(c)$ \\

    \small $\mathindent erc = erc_{tmp} \cup erc$
  };

  \draw[decorate, decoration = {brace, amplitude=5pt}] (18,-9.5) --  (18,-10);
  \node[note, align = left, rotate=90] at (18.7,-8.87) {\textsc{mergeERCSet}$(erc_{tmp}, erc)$};

\end{tikzpicture}

\end{adjustbox}

\caption{Protocol workflow for RRPs and pseudocode for each operations.}
\label{fig:workflow}

\end{figure*}
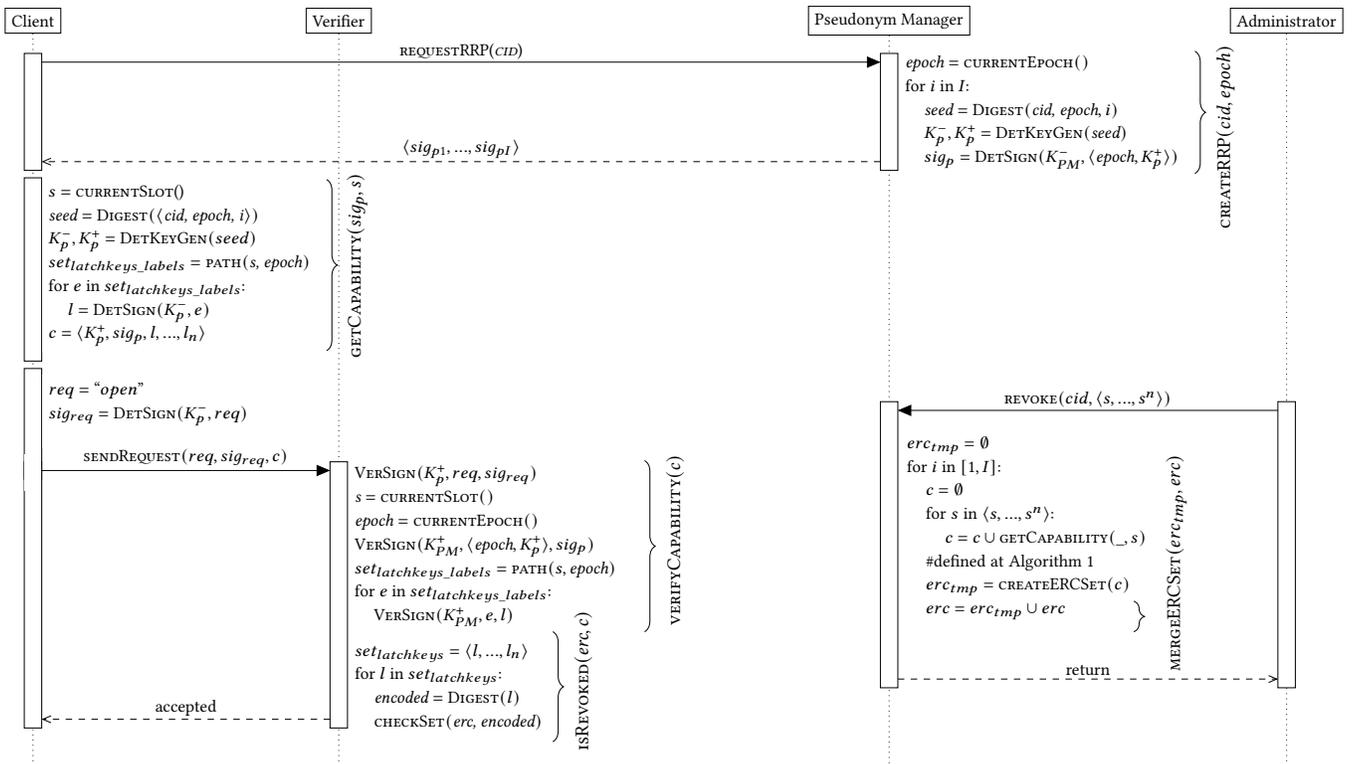

We now present in Figure~\ref{fig:workflow} our protocol and workflow in detail for each operation presented in Section~\ref{sec:rrps}, and how each entity interacts in the system.

\end{document}